\documentclass[
  journal=psrm,
  manuscript=research-note,  
  year=2021,
  volume=2,
]{cup-journal}

\usepackage{booktabs,microtype,siunitx}
\usepackage{multicol}
\usepackage{amsfonts,amsmath}
\usepackage{amssymb}
\usepackage{multirow}
\usepackage{tikz}
\usepackage{tcolorbox}
\usepackage{wrapfig}
\usepackage{amsthm}
\newtheorem{theorem}{Theorem}
\newtheorem{remark}{Remark}

\newtheorem{lemma}{Lemma}
\newtheorem{corollary}{Corollary}
\newtheorem{assumption}{Assumption}

\usepackage{bbm}
\usepackage{amssymb}
\usepackage{comment}
\usepackage{wrapfig}
\usepackage{subfig}
\usepackage{multirow}
\usepackage{tikz,times}
\usepackage{fancybox}
\usepackage{xcolor}
\usepackage{amsmath}
\usepackage{amssymb}
\usepackage{amsthm}
\usepackage{algorithm}
\usepackage{algpseudocode}
\usepackage{bbm}
\DeclareMathOperator{\st}{s.t.}

\newtcolorbox{myorangebox}{colframe = myorange}
\newtcolorbox{mybluebox}{colframe = myblue}
\usetikzlibrary{decorations.pathreplacing,calligraphy,backgrounds}
\usetikzlibrary{positioning}
\usetikzlibrary{mindmap,backgrounds}
\usetikzlibrary{arrows,automata}
\definecolor{myblue}{RGB}{0 114 199}
\definecolor{mylightblue}{RGB}{77 191 241}
\definecolor{myorange}{RGB}{217 83 25}

\newtcbox{\alertinline}[1][red]
  {on line, arc = 0pt, outer arc = 0pt,
    colback = #1!20!white, colframe = #1!50!black,
    boxsep = 0pt, left = 1pt, right = 1pt, top = 2pt, bottom = 2pt,
    boxrule = 0pt, bottomrule = 1pt, toprule = 1pt}
\newtcolorbox{mybox}{colframe = red!50!black}
\title{Synthetic Principal Component Design:\\
Fast Covariate Balancing with Synthetic Controls}

\author{Yiping Lu}
\affiliation{ICME, Stanford University, CA, USA}
\email{yplu@stanford.edu}

\author{Jiajin Li}
\affiliation{Department of Management Science \& Engineering, Stanford University, CA, USA}
\email{gerrili1996@gmail.com}

\author{Lexing Ying}
\affiliation{ICME, Stanford University, CA, USA}
\alsoaffiliation{Department of Mathematics, Stanford University, CA, USA}
\email{lexing@stanford.edu}

\author{Jose Blanchet}
\affiliation{ICME, Stanford University, CA, USA}
\alsoaffiliation{Department of Management Science \& Engineering, Stanford University, CA, USA}
\email{jose.blanchet@stanford.edu}


\begin{document}

\begin{abstract}
The optimal design of experiments typically involves solving an NP-hard combinatorial optimization problem. In this paper, we aim to develop a globally convergent and practically efficient optimization algorithm. Specifically, we consider a setting where the pre-treatment outcome data is available and the synthetic control estimator is invoked. The average treatment effect is estimated via the difference between the weighted average outcomes of the treated and control units, where the weights are learned from the observed data. {Under this setting, we surprisingly observed that the optimal experimental design problem could be reduced to a so-called \textit{phase synchronization} problem.} We solve this problem via a normalized variant of the generalized power method with spectral initialization. On the theoretical side, we establish the first global optimality guarantee for experiment design when pre-treatment data is sampled from certain data-generating processes. Empirically, we conduct extensive experiments to demonstrate the effectiveness of our method on both the US Bureau of Labor Statistics and the Abadie-Diemond-Hainmueller California Smoking Data. In terms of the root mean square error, our algorithm surpasses the random design by a large margin.
\end{abstract}

\section{Introduction}

Estimating the average effects of a binary treatment is one of the main goals of empirical economic and political studies. Randomization in controlled trials is one of the golden rules for estimating average treatment effects (ATE). Suppose the treatment assignment procedure guarantees that the potential outcomes are independent of the treatment status. In that case, a simple difference-in-mean (i.e., average outcomes of the treated and control units) estimator becomes an unbiased estimator of ATE. Nevertheless, a fully randomized experiment may be affected by a significantly high variance in the final estimation. Such variance can be reduced via exploiting the feature information in the observed data. Naturally, we focus on the following question:  \emph{can the observed covariates improve the statistical properties of the ATE estimators via experimental designing?} \citep{rubin2008objective,kasy2016experimenters}. This problem is referred to as \emph{covariate balancing}, which restricts the randomization to achieve covariate balance between treatment groups \citep{efron1971forcing,morgan2012rerandomization}.  

Covariate balancing has been substantially explored in the literature. 
One approach is applying covariate balancing via an adequately designed propensity score~\citep{imai2014covariate,zhao2019covariate}, which requires that we have access to a reasonably large sample of experimental units. However, large samples of experimental units are not always available in practice. Therefore, if only a moderate number of units are accessible to the treatment, \citep{bansal2018gram} solve an NP-hard combinatorial optimization problem to balance the empirical covariates. To address this issue, we aim to design a computationally feasible optimal design of experimental studies. 
%

We specifically consider the optimal design of experiments for the synthetic control estimator as in ~\citep{abadie2003economic,abadie2010synthetic,abadie2015comparative}, which becomes an attractive estimation procedure when only a small number of units can be exposed to the experiment. That is, the experiment designer can observe the pre-treatment panel outcome data for a number of units in a number of time periods. Synthetic control compares treated units with a weighted average of untreated units. The weights are determined via empirical fit on the observed pre-treatment outcome. As an example, consider the study of Proposition 99's effect presented in \citep{abadie2010synthetic}. Proposition 99 is a large-scale anti-smoking legislation program that California implemented in 1988. The policy maker wants to estimate the effect of this piece of legislation. Synthetic control suggests the policy maker to estimate the counterfactual outcomes after 1988 by using the observed outcomes of the states without legislative restrictions. \citep{abadie2010synthetic} produces the synthetic control by a combination of Colorado, Connecticut, Montana, Nevada, and Utah and find out that annual per-capita cigarette sales in California were about 26 packs lower than what they would have been by the year 2000. Beyond this application, the method of synthetic controls has been used in many other empirical policy evaluation problems, including legalized prostitution \citep{cunningham2018decriminalizing}, corporate political connections \citep{acemoglu2016value}, taxation \citep{kleven2013taxation}, to just name a few. 
To realize the benefits offered by synthetic control,  \citep{doudchenko2021synthetic,abadie2021synthetic} considers how an optimally designed experiment can help the experiment designer to further reduce the variance. \citep{doudchenko2021synthetic,abadie2021synthetic} propose an optimization approach to select the control group based on the observed pre-treatment outcome. Namely, the choice of the treated units aims to balance the weighted average of treated and untreated covariates. 
As such, the designer can choose the best non-negative weights. \citep{doudchenko2021synthetic} proves that the underlying optimization problem is still NP-hard and \citep{abadie2021synthetic} relaxes the optimization problem into a canonical Quadratic Constraint Quadratic Program
(QCQP). Nevertheless, the resulting QCQP is rather computationally demanding and applicable algorithms (e.g., SDR(cf. Semidefinite Programming)~\citep{luo2010semidefinite}) are not guaranteed to reach a global optimum.

In this paper, we aim to design the first globally convergent optimization algorithm for the weighted covariate balancing formulation \citep{doudchenko2021synthetic}. Moreover, the algorithm is practically efficient. We achieve this by first removing the so-called units cardinality constraint being treated in \citep{doudchenko2021synthetic}. Surprisingly, we find out that this relaxation can be shown to be equivalent to a phase synchronization problem \citep{singer2011angular}. Although phase synchronization is still an NP-hard problem \citep{zhang2006complex}, many practical algorithms have been recently developed for this nonconvex problem, \citep{bandeira2017tightness,boumal2016nonconvex,liu2017estimation,zhong2018near}. Moreover, as we will argue, this nonconvex problem is polynomial-time solvable under a suitable data generating process (\emph{i.e.} average case complexity). Motivated by this line of research, we propose \emph{Synthetic principal component Design} (SPCD), which optimizes the treatment decision via (a normalized variate of) the generalized power method with spectral initialization \citep{chen2021spectral}. If we assume that the pre-trement data is sampled from the linear fixed-effect model studied in \citep{abadie2010synthetic,xu2017generalized,athey2021matrix,ferman2021properties} and  invoke the realizable assumptions considered in  \citep{abadie2021synthetic}, we can further establish its \emph{global} optimization guarantee  and statistical estimation guarantee for the proposed synthetic control procedure. To the best of our knowledge, this is the first computational paradigm of combinatorial optimization-based experiment design which enjoys a global optimization guarantee.

\paragraph{Paper organization} We organize our paper as follows. In Section \ref{section:setting}, we introduce the setup in \citep{doudchenko2021synthetic,abadie2021synthetic} where the authors consider
covariate balancing under the synthetic control setting and reformulate it as a phase
synchronization problem. In Section \ref{section:GIPW}, we introduce the generalized power
method and provide global a optimization guarantee under the linear factor model
\citep{abadie2003economic,xu2017generalized,athey2021matrix,ferman2021properties} with a realizable
assumption \citep{abadie2021synthetic}. In Section \ref{section:numerical}, we apply our method to
both simulated and real-world data sets. We end with some closing remarks in Section \ref{section:conl}.

\vspace{-0.1in}
\subsection{Related Work}

\paragraph{Synthetic Control} Synthetic control \citep{abadie2003economic,abadie2010synthetic} is one of the leading methods to estimate the causal effect of a binary treatment in the panel data setting. It constructs a weighted combination of untreated groups used as controls, to which the outcome of the treatment group is compared. This construction of the control group significantly improves the performance when there are limited units that can be exposed to the experiment designer.
Statistical consistency and inference properties have been established for the synthetic control under the linear factor models \citep{xu2017generalized,athey2021matrix,ferman2021properties}, which is also the starting point of our theory.

\paragraph{Covariate Balancing} To make the difference-in-mean estimator more precise, experimenters sometimes restrict the treatment assignment to achieve covariate balance between treatment groups, \emph{i.e.,}
\[
\min_{\{D_i\}_{i=1}^n} \left |\left|\frac{1}{\sum
\limits_{i}D_i}\sum_{i:D_i = 1}X_{i} -\frac{1}{\sum
\limits_{i}(1-D_i)}\sum_{i:D_i = 0}X_{i}\right |\right|^2
\]
where $\{X_i\in\mathbb{R}^n\}_{i=1}^n$ are the observed units and $\{D_i\in\{0,1\}\}_{i=1}^n$
denotes the treatment experiments the designer aims to optimize
\citep{efron1971forcing,morgan2012rerandomization,imai2014covariate,harshaw2019balancing}.  Indeed,
this is a 0-1 NP-hard combinatorial optimization problem. Various different methods to approximately
solve the problem have been proposed, including rerandomization
\citep{morgan2012rerandomization,kallus2018optimal,kallus2021optimality,li2018asymptotic}, design
particular propensity score \citep{imai2014covariate,zhao2019covariate} and the recently proposed
Gram-Schmidt walk \citep{bansal2019algorithm,harshaw2019balancing}. In this paper, we follow the setting in
\citep{doudchenko2021synthetic,abadie2021synthetic} which extend the covariate balancing problem to the
synthetic control setting. \citep{doudchenko2021synthetic} has proved that this problem is still
NP-hard and \citep{abadie2021synthetic} further relaxes the optimization problem into a canonical
form of QCQP.

\paragraph{Phase Synchronization} Phase synchronization aims to estimate $n$ unknown angles $\theta_i\in [0,2\pi],i\in[N]$ through noisy measurements of their offset $\theta_i-\theta_j$ mod $2\pi$ \citep{singer2011angular}, which has received intense interests in areas such as time synchronization between distributed networks, \citep{giridhar2006distributed}; ranking \citep{cucuringu2016sync}; computer vision \citep{wang2013exact,martinec2007robust}; and optics and inverse problems \citep{rubinstein2001reconstruction,alexeev2014phase,singer2011three}. 
To globally address the non-convex optimization problem arising from the inverse problem,  \citep{bandeira2017tightness} provides the first global results for the SDR  under certain data generating procedures. In this paper, we follow a line of research which takes advantage of generalized power methods \citep{boumal2016nonconvex,liu2017estimation,zhong2018near} to solve the problem in a way that shares a similar global guarantee to that of the SDP relaxation, but is computationally much more amenable. It is worth noting that, although the optimization problem considered in our paper is the same as phase synchronization, the data generating process is different from the considered earlier and this leads to our problem-specific generalized power method. 

\subsection{ Main Contributions}

\begin{itemize}
    \item We show a surprisingly equivalence between the experiment design with synthetic control \citep{doudchenko2021synthetic,abadie2021synthetic} and phase synchronization problem \citep{singer2011angular,boumal2016nonconvex,liu2017estimation}, where separating the experiment and control group can be transformed to finding the phase of a complex signal. We further reveal the hidden connection between covariate balancing with the smallest eigenvector of the gram matrix 
    and propose a spectral method for fast experiment design.
    \item We propose a novel normalized version of the generalized power method, which enjoys \emph{global} convergence results under certain generative models. The normalization technique can weak the condition for generative model assumptions to guarantee the global optimiality and also consistently improves the empirical results.
    \item In terms of  the root mean square error, our method empirically surpasses random design by a large margin on both synthetic and real-world datasets. Our performance can even exceed 500000 times of rerandomization over the Abadie-Diamond-Hainmueller smoking legislation data. 
\end{itemize}

\section{Mathematical Formulations}
\label{section:setting}
In this section, we  introduce the mathematical formulation of synthetic control (SC) and the optimal experiment design problem studied in \citep{doudchenko2021synthetic}.

\paragraph{Problem setup} We aim to estimate the effect of a binary treatment under the panel data setting. Researchers have access to the outcome metric of interest $Y \in\mathbb{R}^{N\times T}$ for $N$ units during $T$ time periods. At time $T$, researchers are required to execute an experiment by assigning a binary treatment described by $D_i\in\{-1,1\}, i\in [N]$ based on the observed pre-treatment data. 
 If $D_i=1$, then a treatment needs to be applied to unit $i$. After the treatment experiment, furthermore outcomes are observed for additional $S$ time periods $t= T+1,\cdots, T+S$. During this period, every unit $i\in[N]$ in each time period $t$ is associated with the following two random outcomes: $Y_{it}(-1)=\mu_{it}+e_{it}, \text{ and } Y_{it}(1)=Y_{it}(-1)+\tau_i,$ where $\mu_{it}$ is the base outcome, ${\tau}$ is the treatment effect aiming to estimate and  $e_{it}$ is the zero mean i.i.d idiosyncratic noise with variance $\text{Var}(\epsilon_{it})=\sigma$. Once treatment $D_i$ is applied, the experimenter is able to realize $Y_{it}=\frac{(D_i+1)}{2}Y_{it}(1)+\frac{(1-D_i)}{2}Y_{it}(-1)$.

Estimating the treatment effect $\tau$ is quite challenging because once we implement a treatment on unit $j$ (\emph{i.e.,} $D_j=1$) and observe the outcome $Y_{j,T+1}(1)$, then counterfactual outcome $Y_{j,T+1}(-1)$ is not observable. With the pre-treatment observation $Y_{iT}$, synthetic control literature \citep{abadie2010synthetic,xu2017generalized} constructs the counterfactual estimate for a treated unit $j$  from a weighted average of other units' outcomes: $\hat Y_{j,T+1}(-1)=\sum_{i:D_i=-1} w_iY_{i,T+1}$. The weights $w_i$ are learned from the pre-treatment observed data via minimizing $\sum_{t=1}^T(Y_{jt}-\sum_{i:D_i=0}w_iY_{it})^2$. Then the treatment effect of unit $j$ we estimate can be written as $\tau_j=Y_{j,T+1}-\hat Y_{j,T+1}(-1)$. 
\subsection{Synthetic Design}

In this subsection, we consider the synthetic design objective function proposed for two-way fixed effect in \citep{doudchenko2021synthetic,abadie2021synthetic} and reveal its hidden connection with the phase synchronization problem. \citep{doudchenko2021synthetic,abadie2021synthetic}  aim to design treatment assignments $\{D_i=\pm 1\}_{i=1}^N$ and weights $\{w_i\ge0\}_{i=1}^N$ for outcome experiments at time $T+1$. For we aims to estimate a two-way fixed effect where the treatment effects are homogeneous, we can consider the \emph{weighted average treatment effect on the treated (wATET)} $\tau=\sum_{i=1}^N \frac{D_i+1}{2}w_i\tau_i$ instead \citep{bottmer2021design}. wATET can be estimated as a difference in weighted means estimator $\hat\tau=\sum_{i:D_i=1}w_i Y_{i,T+1}-\sum_{i:D_i=-1}w_i Y_{i,T+1}$ with $\sum_{i:D_i=1}w_i=\sum_{i:D_i=-1}w_i=1$. Following \citep{doudchenko2021synthetic}, upon the outcome model, the mean squared error of the difference-in-weighted-means estimator admits the decomposition
$$
\mathbb{E}\left[(\hat \tau-\tau)^2|\{D_i,w_i\}_{i=1}^N\right]=\underbrace{\left(\sum_{i:D_i=1} w_i\mu_{i,T+1}-\sum_{i:D_i=-1}w_i\mu_{i,T+1}\right)^2}_{\text{weighted covariate balancing}}+\sigma\sum_{i=1}^N w_i^2.
$$
The designer aims to design the experiment with a lowest expected mean square error. Thus,  \citep{doudchenko2021synthetic} proposed the following mixed-integer programming for experimenting design with Synthetic Control: 
\begin{equation}
    \begin{aligned}
    & \min_{\{D_i,w_i\}_{i=1}^n} \, \frac{1}{T}\sum_{t=1}^T\left(\sum_{i:D_i=1} w_iY_{it}-\sum_{i:D_i=-1}w_iY_{it}\right)^2+\sigma \sum_{i=1}^Nw_i^2\\
    & \quad \, \st \quad \,\,\,\,  w_i\ge 0, D_i\in\{-1,1\}, \,\,  \forall i \in [N], \\
    & \quad \quad \quad \quad   \sum_{i:D_i=1} w_i=\sum_{i:D_i=-1} w_i=1.
    \end{aligned}
    \label{eq:originalSD}
\end{equation}
\begin{remark} We remove the constraint $\sum_{i:D_i=1} D_i=K$ for a given integer $K\in\mathbb{N}$ in the mixed-integer programming in \citep{doudchenko2021synthetic} mainly as this constraints is empirically proved not critical in \citep{abadie2021synthetic}. The NP-hard proof in \citep{doudchenko2021synthetic} depends on the constraint $\sum_{i:D_i=1} D_i=K$. In the following discussion, we will show that the problem is also NP-hard even $\sum_{i:D_i=1} D_i=K$ is removed, as the resulting optimization problem can be reformulated as the $\ell_1$-PCA \citep{mccoy2011two,wang2021linear} and phase Synchronization \citep{singer2011angular,boumal2016nonconvex} problems.
\end{remark}

By making a further simplification of the problem \eqref{eq:originalSD}, we introduce a change of
variable $W_i=w_iD_i$. For $w_i\ge 0$, then $D_i=\textnormal{sgn}(W_i)$ and $w_i=|W_i|$. At the same
time, the constraint $\sum_{i:D_i=1} w_i=\sum_{i:D_i=-1} w_i=1$ is equivalent to $\mathbbm{1}^\top
W=0$ and the objective function
\[
\frac{1}{T}\sum_{t=1}^T\left(\sum_{i:D_i=1} w_iY_{it}-\sum_{i:D_i=-1}w_iY_{it}\right)^2+\sigma \sum_{i=1}^Nw_i^2
\]
can be reformulated as $W^\top (YY^\top+\lambda I) W$, where $W=[w_1,\cdots,w_N]^\top$ and
$\mathbbm{1}\in\mathbb{R}^N$ is the all one vector. Thus, \eqref{eq:originalSD} can be
recast into
\begin{equation}
\label{eq:min_l1}
    \begin{aligned}
    \min_{W\in\mathbb{R}^n, \mathbbm{1}^\top W=0, ||W||_1=1 } W^\top (YY^\top+\sigma I) W.
    \end{aligned}
\end{equation}
Although the reformulation \eqref{eq:min_l1} translates the problem into a compact matrix form, it is still a nonconvex problem due to the constraint $||W||_1=1$. To deal with the constraint $\mathbbm{1}^\top
W=0$, we add an extra term $\lambda(\mathbbm{1}^\top W)^2$ to the objective function, where $\lambda$ is a pre-defined hyper-parameter. Although this penalty
method cannot produce the exact global solution, we can still recover the sign of the global
solution (see Theorem \ref{thm:sign}). Once the sign of the global solution is identified, the remaining effort of computing the magnitude reduces to solving a convex problem (\ref{eq:convexdesign}).

\begin{theorem} 
For large enough $\lambda$, the global solution $W^\ast$ of (\ref{eq:min_l1}) satisfies
$$
\textnormal{sgn}(W^\ast)=\textnormal{sgn}\left(\mathop{\arg\min}_{W\in\mathbb{R}^n, ||W||_1=1 } W^\top (YY^\top+\sigma I+\lambda \mathbbm{1}\mathbbm{1}^\top) W\right).
$$
\label{thm:sign}
\end{theorem}
The following theorem states that the problem is equivalent to another well-known NP-hard non-convex
problem --- Phase Synchronization \citep{singer2011angular}.

\begin{theorem}[Equivalence between Synthetic Design and Phase Synchronization] If $x^\ast\in\mathbb{R}^n$ is the global solution of $\min_{||x||_1=1} ||Ax||_2^2$ for some matrix $A\in\mathbb{R}^{D\times n}$ ($D>n$) and the matrix $A^\top A\in\mathbb{R}^{n\times n}$ is invertible, then $y^\ast=\textnormal{sgn}(x^\ast)$ is the global solution of $\max_{y\in\{-1,+1\}^n} y^\top((A^\top A)^{-1})^\top y$.
\label{theorem:equaltoPhaseSynch}
\end{theorem}
The proof of Theorem \ref{thm:sign} and Theorem \ref{theorem:equaltoPhaseSynch} is omitted in the main text due to page limit and is shown in Appendix \ref{section:equal}.
\begin{remark} Phase synchronization \citep{singer2011angular,bandeira2017tightness} aims to recover $n$ phases $z_i=e^{i\theta_i}, i\in[n]$  via solving the following optimization problem
\begin{equation}
    \begin{aligned}
    \max_{|x_1|=\cdots=|x_n|=1}\qquad x^\top C x,\label{eq:phaseSyn}
    \end{aligned}
\end{equation}
where $C_{ij}$ is the noisy observation of $z_i\bar{z_j}=e^{i(\theta_i-\theta_j)}$. Our problem is symbolically equivalent to (\ref{eq:phaseSyn}). However, the data generating process is quite different from the Phase synchronization. The matrix we consider is the inverse of the data matrix. In Appendix, we will show that the design we find is actually the first $\ell_1$-principal component \citep{mccoy2011two,wang2021linear}.
\end{remark}

\section{Algorithm description}
\label{section:GIPW}
In this section, we propose the generalized power method with spectral initialization to solve our problem. The method is inspired by the lines of its early efforts for phase synchronization problems \citep{journee2010generalized,boumal2016nonconvex,liu2017estimation,zhong2018near}.

\begin{algorithm}[t]
\caption{\textbf{\underline{Synthetic principal component Design}}}\label{alg:sytheticPCD}
\begin{algorithmic}
\Require Pre-treatment Observations $Y\in\mathbb{R}^{N\times T}$
\State Set initial treatment assignment guess through $y^0=\textnormal{sgn}(v)$, where $v$ is the smallest eigenvector of matrix $(Y Y^\top+\alpha I+\lambda \mathbbm{1}\mathbbm{1}^\top)$,  where $\alpha,\lambda$ are two pre-defined hyper-parameter.
\State\Comment{\textbf{Spectral Initialization}}
\While{Converged}
\State Select one of the following two boxes to iterate
\begin{myorangebox}
\State For SPCD, update the design via \Comment{\textbf{Generalized power methods}} 
\begin{equation}
    \begin{aligned}
    {\color{myorange}y^{t+1}=\textnormal{sgn}\left[\left((YY^\top+\alpha I+\lambda \mathbbm{1}\mathbbm{1}^\top)^{-1}+\beta I\right) y^{t}\right]},
    \end{aligned}
\end{equation}
\State where $\beta$ is a pre-defined hyper-parameter.
\end{myorangebox}
\begin{mybluebox}
\State For NormSPCD, update the design via \Comment{\textbf{Normalize the inverse covariance matrix}}
\begin{equation}
    \begin{aligned}
    {\color{myblue}y^{t+1}=\textnormal{sgn}\Big[\left[(Y Y^\top+\alpha I+\lambda \mathbbm{1}\mathbbm{1}^\top)^{-1}+\beta I\right] (y^{t}/d)\Big]},
    \end{aligned}
    \label{update:normsgd}
\end{equation}
\State where {$d=\sqrt{\textnormal{diag}((Y Y^\top+\alpha I+\lambda \mathbbm{1}\mathbbm{1}^\top)^{-1})}$} and $/$ denotes element-wise divide.
\end{mybluebox}
\EndWhile
\State Solve the following \emph{convex} optimization problem 
\begin{equation}
    \begin{aligned}
    \{w_i\}_{i=1}^n=& \mathop{\arg\min}_{\{w_i\}_{i=1}^n} \,\, \frac{1}{T}\sum_{t=1}^T\left(\sum_{i:y(i)=1} w_iY_{it}-\sum_{i:y(i)=-1}w_iY_{it}\right)^2+\sigma \sum_{i=1}^Nw_i^2\\
    & \quad \, \st \quad \,\,  w_i\ge 0,  \,\,  \forall i \in [N],  \sum_{i:y(i)=1} w_i=\sum_{i:y(i)=-1} w_i=1.
    \end{aligned}
    \label{eq:convexdesign}
\end{equation}
\State Treat Unit $i$ if $y(i)=-\textnormal{sgn}\left(\sum_{i=1}^N y(i)\right)$ and run the experiment. 
\State\Comment{\textbf{To ensure the size of the treated group is smaller than the control group}} 
\State Estimate the treatment effect via
$$
\hat\tau=\sum_{t=1}^S\left(\sum_{i:y(i)=-\textnormal{sgn}\left(\sum_{i=1}^N y(i)\right)}w_iY_{i,T+t}-\sum_{i:y(i)=\textnormal{sgn}\left(\sum_{i=1}^N y(i)\right)}w_iY_{i,T+t}\right).
$$
\end{algorithmic}
\end{algorithm}

\vspace{-0.1in}
\subsection{Generalized Power Methods}

Spectral relaxation \citep{singer2011angular} is the first simple and efficient approach to solve the
phase synchronization problem. \citep{singer2011angular} relaxed the $N$ constraints
$|x_i|=1,i\in[N]$ to $||x||_2^2=n$. Then the solution becomes the leading
eigenvector. \citep{liu2017estimation,zhong2018near} showed that the eigenvector estimator is almost
close to the global optima under certain data generating process. Following these works, we take our
initial guess of the optimal experiment to be $\textnormal{sgn}(v)$, where $v$ is the smallest
eigenvector of the matrix $(Y Y^\top+\alpha I+\lambda \mathbbm{1}\mathbbm{1}^\top)$ with
$\alpha,\lambda>0$ as two pre-defined hyper-parameters.

To further improve the experiment assignment, we utilize the generalized power method
\citep{journee2010generalized,luss2013conditional}, which considers the linearization of the
objective function at the current point and moves towards a minimizer of this linear function over
the non-convex set $\mathcal{C}$. It is also referred as to Frank-Wolfe algorithm for non-convex
problems in the literature.
Thus, the generalized power method enjoys monotonic improvement (\citep[Lemma 8]{boumal2016nonconvex}). When $g(x)$ is a quadratic function taking the form as $x^\top A x$, the method will become similar to the power methods for the eigenvector problem. The only difference is the normalization step. The power method normalizes the whole vector but the generalized power method normalizes the individual entries. The generalized power method can be also understood as projected gradient descent \citep{smith1994optimization}. Indeed, the update 
\[
y^{t+1}=\textnormal{sgn}[((\frac{1}{\beta} YY^\top+\frac{\sigma}{\beta} I+\frac{\alpha}{\beta } \mathbbm{1}\mathbbm{1}^\top)^{-1}+I) y^{t}]
\]
can be understood as a projection step ($\textnormal{sgn}$) after a gradient descent update with step size 
$\frac{1}{\beta}$:
\[
y^{t+1}=\textnormal{sgn}(y^{t}+\frac{1}{\beta}((YY^\top+\sigma I+\alpha \mathbbm{1}\mathbbm{1}^\top)^{-1})y^{t}).
\]
Thus the algorithm shares a sufficient ascent condition for each iteration (shown in Appendix
\ref{subsection:linearrate}). Our algorithm is called Synthetic principal component Design (SPCD) and
it is summarized in Algorithm \ref{alg:sytheticPCD}.

\paragraph{Normalized Variant}
In \citep{boumal2016nonconvex,zhong2018near,liu2017estimation}, the global optimality result is
highly dependent on the assumption that the top eigenvector of the iteration matrix lies in
$\{-1,1\}^N$. However, in our setup, the top eigenvector of the iteration matrix is the smallest
eigenvector of the covariance matrix which may not be a $\{-1,1\}^N$ vector. This is also the case
appearing in the phase retrieval \citep{candes2015phase,chen2015solving} and the degree corrected
stochastic block model \citep{zhao2012consistency,jin2015fast}. Inspired by the SCORE
\citep{jin2015fast,jin2022improvements} method for degree corrected stochastic block models, we
further introduce a normalization step to the generalization power method and call the new algorithm
Normalized SPCD (cf. NormSPCD, see \eqref{update:normsgd} for details). We use the diagonal
component of the inverse covariance as an estimate of the true normalization component. In Appendix
\ref{section:opt}, we show that NormSPCD can be interpreted as a Riemannian gradient descent with a
specific metric. Empirical results show that it is better than the original GPW in Figure
\ref{subfig:differentT}. This normalization technique may be of independent interest in other
applications.

\vspace{-0.1in}
\subsection{Global Guarantee}

In this subsection, we provide the global optimization guarantee for the (normalized) generalized power
method. \citep{bandeira2017tightness,boumal2016nonconvex,liu2017estimation,zhong2018near} have shown
that phase retrieval is globally solvable under certain generative models. We will show that
generalized power method can globally converge under certain data generating processes, which are
quite different from the ones assumed in the previous works. Following \citep{abadie2021synthetic},
we consider a realizable linear factor model (also referred to as "interactive fixed-effects model")
\citep{abadie2010synthetic,xu2017generalized,athey2021matrix,ferman2021properties}, which has already
been commonly employed in the literature as a benchmark model to analyze the properties of synthetic
control estimators \citep{amjad2018robust,li2020statistical}. Recently, \citep{shi2021assumptions}
justify the linear assumption from an independent causal mechanism viewpoint. The linear latent
factor model is stated in the following assumption.  .

\begin{assumption}[Linear Latent Factor Model \citep{abadie2010synthetic,xu2017generalized}]
  The outcomes are generated via the following linear factor model
  \[
  Y_{jt}=\delta_t+\frac{D_{jt}+1}{2}\tau+\theta_t^T \mu_j+e_{jt},
  \qquad \mathbb{E}[e_{jt}|\delta_t,\mu_j,D_{jt}]=0,
  \qquad \textnormal{Var}[e_{jt}|\delta_t,\mu_j,D_{jt}]=\sigma.
  \]
  Here $\delta_t$ is the time fixed effect; $\mu_j$ is the unobserved common factors; $\theta_t$ is a
  vector of unknown factor loading; $e_{jt}$ is the unobserved i.i.d. idiosyncratic noise; $\tau$ is the
  treatment effect that we aim to estimate and $D_{jt}$ is the $\{-1,1\}$ variable according to the treatment
  assignment to unit $j$ at time $t$. More specifically, in the pre-treatment period, $D_{jt}=-1$ for
  all $\forall j\in[N], t\in [T]$.
  \label{assumption:factormodel}
\end{assumption}

To obtain the global optimality result, we further make the following realizable assumption that
there is only one realizable experiment (zero error experiment) in population.

\begin{assumption} [Realizable Assumption]
There exists a unique parameter $(w_i,D_i)_{i=1}^n$ that satisfies the following conditions:
\begin{itemize}
\setlength{\itemsep}{0pt}\setlength{\parsep}{0pt}\setlength{\parskip}{0pt}
    \item $D_i$ are binary treatments, \emph{i.e.} $D_i\in\{-1,1\}$.
    \item $w_i\ge 0$  and $\sum_{i=1}^n D_iw_i=0$.
    \item $||w||_2^2=N$ and $\epsilon\le|w_i|\le\frac{1}{\epsilon}$ for all $\forall i \in [N]$.
    \item The weights will balance the covariates, \emph{i.e.} $\sum_{i=1}^nw_iD_i\mu_i=0$. 
\end{itemize}
\label{assumption:realizable}
\end{assumption}

\begin{remark}
  This realizable assumption is similar to \citep{li2020statistical}, \citep[(5)]{shi2021assumptions}
  and \citep[Assumption 3]{abadie2021synthetic}. The difference is that \citep[Assumption
    3]{abadie2021synthetic} assumes the weight will cancel noisy observation of the untreated
  outcome $Y_{jt}(-1)$ which is not realistic when the pre-treatment period $T$ is larger than the
  number of units $N$. Our assumption is closer to \citep{li2020statistical} and
  \citep[(5)]{shi2021assumptions}, but we further assume the uniqueness of the realizable experiment
  that makes the optimization problem easier (in terms of no need to distinguish different
  realizable experiments).
\end{remark}

Under Assumptions \ref{assumption:factormodel} and \ref{assumption:realizable}, we can show the following global optimality result and the proof is shown in the Appendix \ref{section:opt}.


\begin{theorem}
  (Informal) Suppose that Assumptions \ref{assumption:factormodel} and \ref{assumption:realizable}
  hold and that the latent time factor $[\theta_i^\top \delta_i]^\top$ is sampled from a underlying
  distribution
  with
  mean $\Tilde{\theta}$ and covariance $\Tilde{\Sigma}$. Under regularity assumptions (see Appendix
  \ref{subsection:generative}), if $\sigma$ is small enough and
  $T\ge\textnormal{poly}(N,\frac{1}{\epsilon})$, then
  \begin{itemize}
    \setlength{\itemsep}{0pt} \setlength{\parsep}{0pt} \setlength{\parskip}{0pt}
  \item If $\epsilon>\frac{\sqrt{3}}{2}-1$, then SPCD converges to the global optima.
  \item If $\epsilon>0$, then NormSPCD converges to the global optima at a linear rate. 
  \end{itemize}
\end{theorem}

\section{Numerical Study}
 \label{section:numerical}

This section report the numerical tests of our algorithm. In subsection \ref{subsection:simulate},
we validate our algorithm on the latent factor model. In subsection \ref{subsection:realworld}, we
demonstrate our algorithm on two real world datasets. Both experiments have shown the effectiveness
of our proposed experiment design algorithm in terms of the the root-mean-square error (RMSE), where
the squared differences between the true values of the treatment effects and the respective
estimates are computed for each treatment period and averaged.

We first introduce a simplified implementation of (Norm)SPCD, which although not guaranteed optimum but efficient, simple and effective in practice. In the simplified implementation, we don't solve the convex program (\ref{eq:convexdesign}) exactly, but using $w=\frac{2(YY^\top+\alpha I+\lambda \mathbbm{1}\mathbbm{1}^\top)^{-1}y^\ast}{||(YY^\top+\alpha I+\lambda \mathbbm{1}\mathbbm{1}^\top)^{-1}y^\ast||_1}$ to approximate instead. From \ref{eq:equalproof}, we know that once the optimal design profile $y^\ast$ is obtained, then $w$ is the optimal design weight. Notice that we don't exactly globally solve the problem (\ref{eq:min_l1}) in the simplified implementation, although we obtained the right experiment profile $y^\ast$ (Theorem \ref{thm:sign}). The weight we obtained here is the solve of the penalized approximation, but empirically it works good. The whole process is described in Algorithm \ref{alg:realalg}. In all the experiment in this paper, we use this simplified implementation.
\begin{figure}[]
    \centering \includegraphics[width=3in]{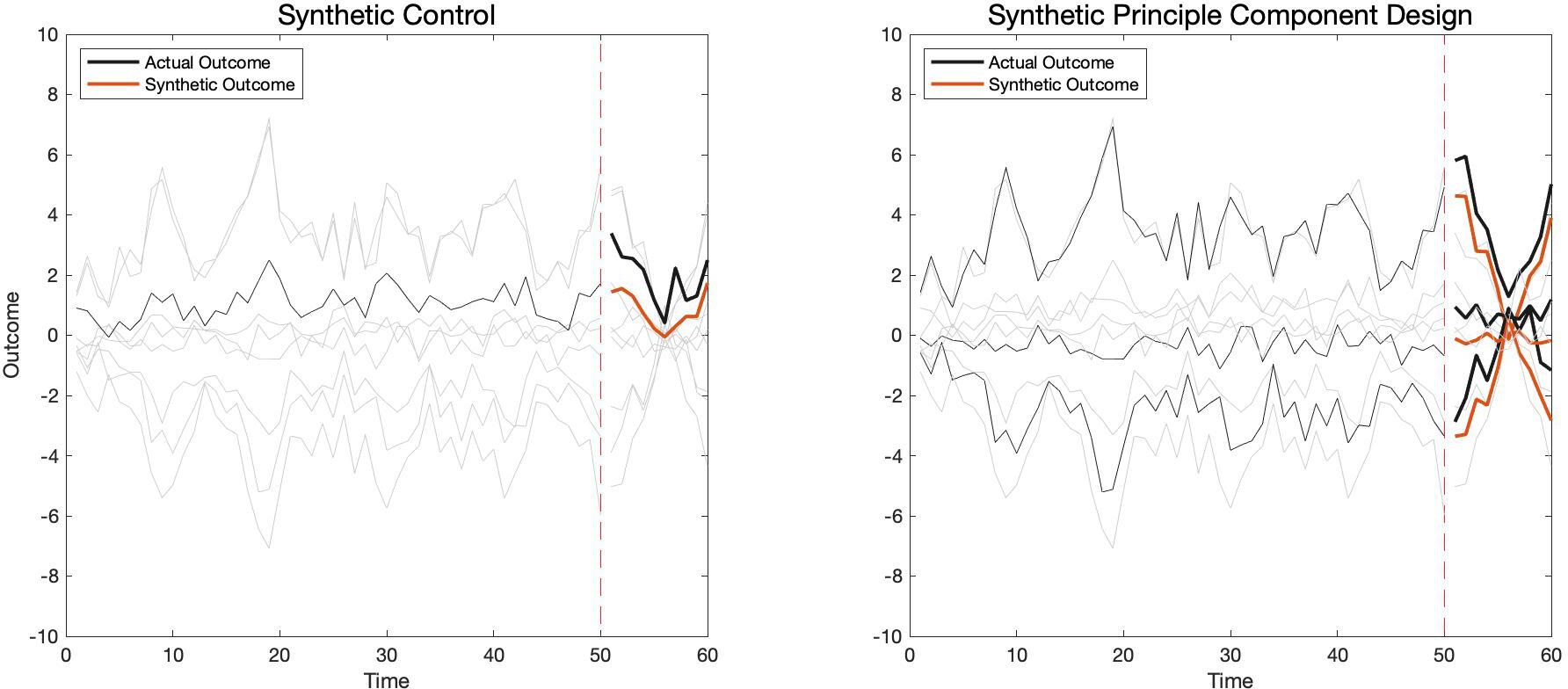}
    \caption{Synthetic principal component design (SPCD) selects the treated units whose features are representative of the whole aggregate market of interest.}

    \label{fig:syndesign_timefixedeffect}
\end{figure}
\vspace{-0.1in}
\subsection{Simulated Data}
\label{subsection:simulate}

We generate data from the linear factor model (also referred to as an "interactive fixed-effects
model", \citep{xu2017generalized,athey2021matrix}). The outcome $Y$ comes from
$$Y_{it}= v_t^\top\gamma_i+\tau \frac{D_{it}+1}{2}+e_{it},\,\forall i\in[N],\forall t\in[T+S].$$ where
$\gamma_i$ is a vector of latent unit factor of dimension $L$ generated as a standard Gaussian. We simulate the treatment effect $\tau$ with ground turth $1$. The idiosyncratic
noise is sampled from Normal $(0,\sigma^2)$ with $\sigma=1$. Fixing the test time period $S=10$ and the number of units $N=10$, we simulate different pairs of $(L,T)$ selection. We simulated three different choices of the time factor vector $v_t$

\begin{figure}
  \centering
  \subfloat[{$L=20,T=9$}]{
        \includegraphics[width=0.18\textwidth]{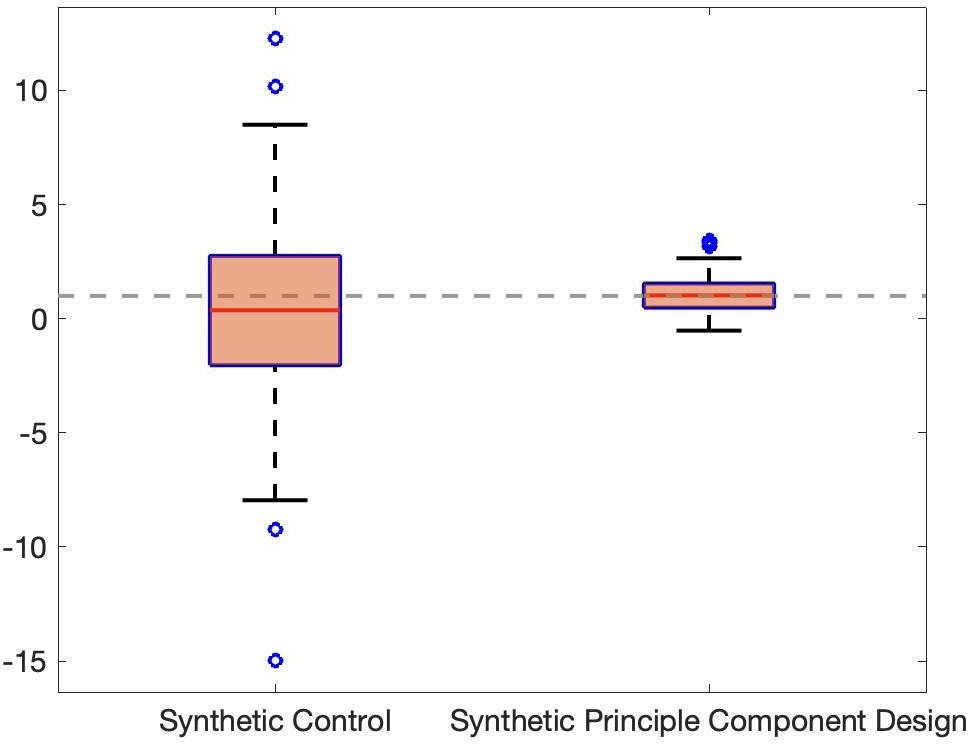}
         \label{subfig:arandom}}
  \quad
  \subfloat[{$L=20,T=20$}]{
        \includegraphics[width=0.18\textwidth]{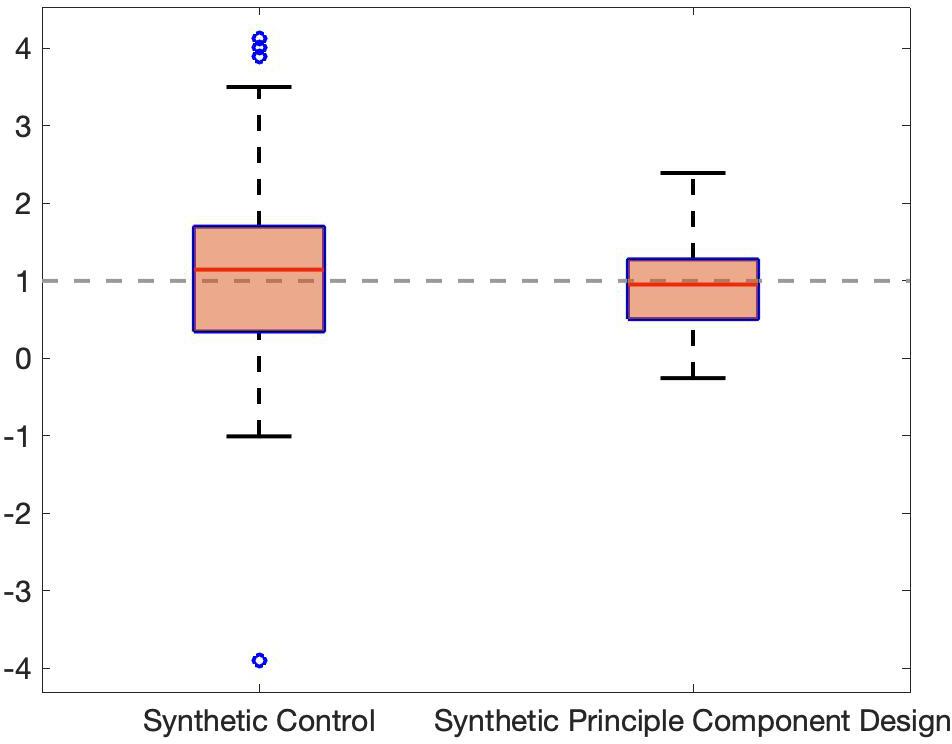}\label{subfig:brandom}}
        \subfloat[{$L=8,T=9$}]{
        \includegraphics[width=0.18\textwidth]{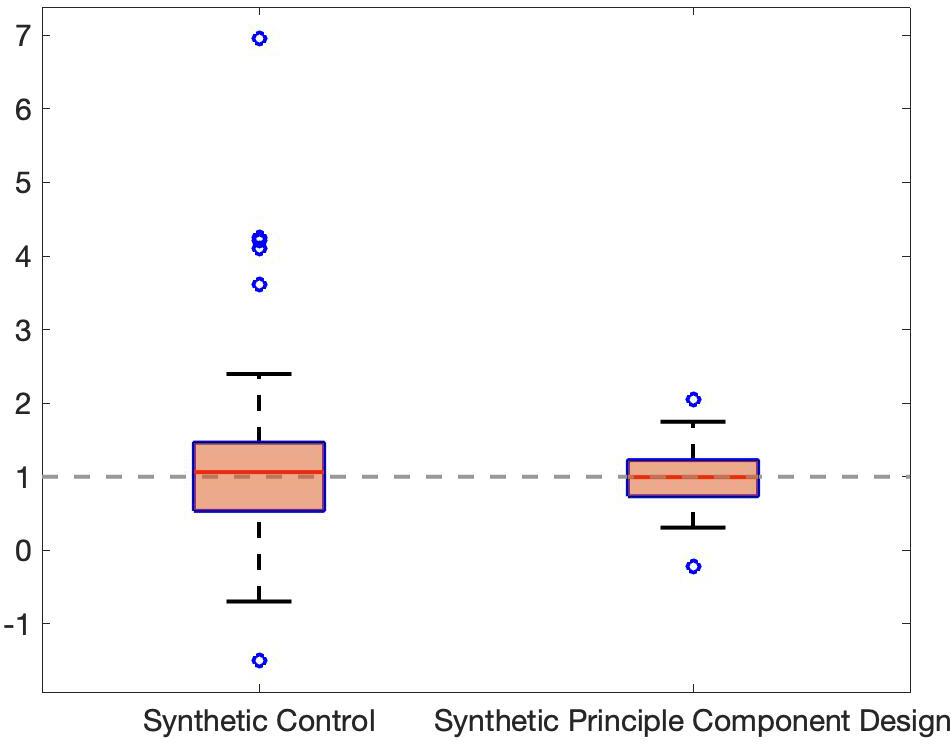}
         \label{subfig:crandom}}
         \quad
         \subfloat[{$L=8,T=20$}]{
        \includegraphics[width=0.18\textwidth]{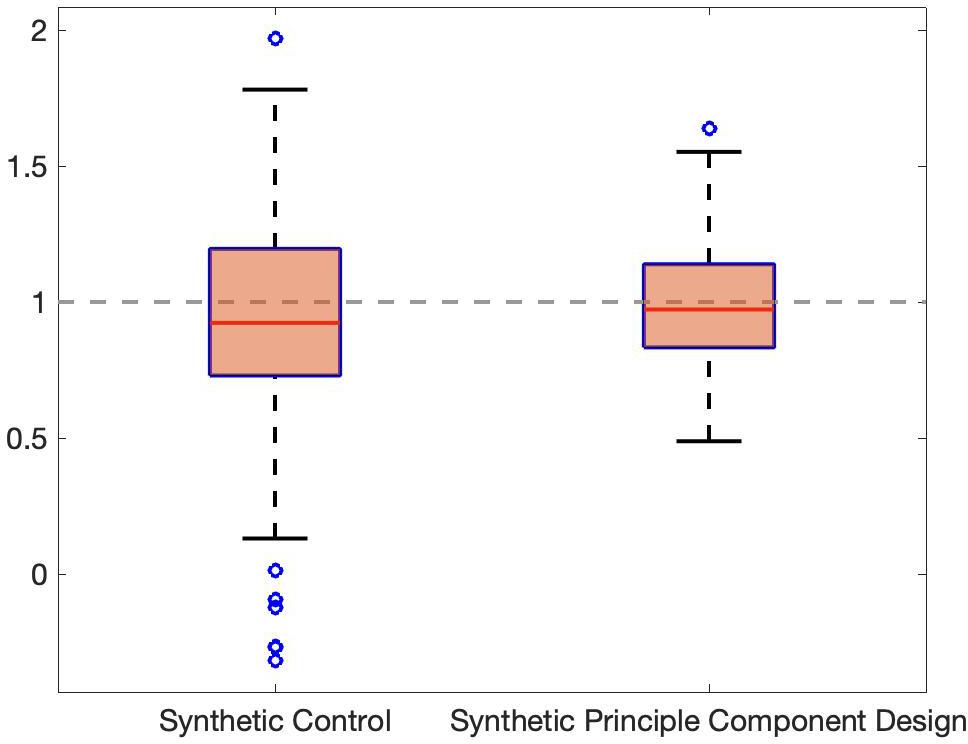}
         \label{subfig:drandom}}
        \caption{Treatment estimated via Synthetic Control and Synthetic Principle Design for data generated from pure random latent vector. We run the experiment over 100 runs of different seeds for different selections of $L,T$ on data generated from purely random latent vector.In all cases, Synthetic Principle Design provides more robust estimate of the true treatment effect 1.}
        \label{figure:simulatedrandom}
 
\end{figure}

\begin{figure}
  \centering
  \subfloat[{\scriptsize$L=20,T=9$}]{
    \includegraphics[width=0.18\textwidth]{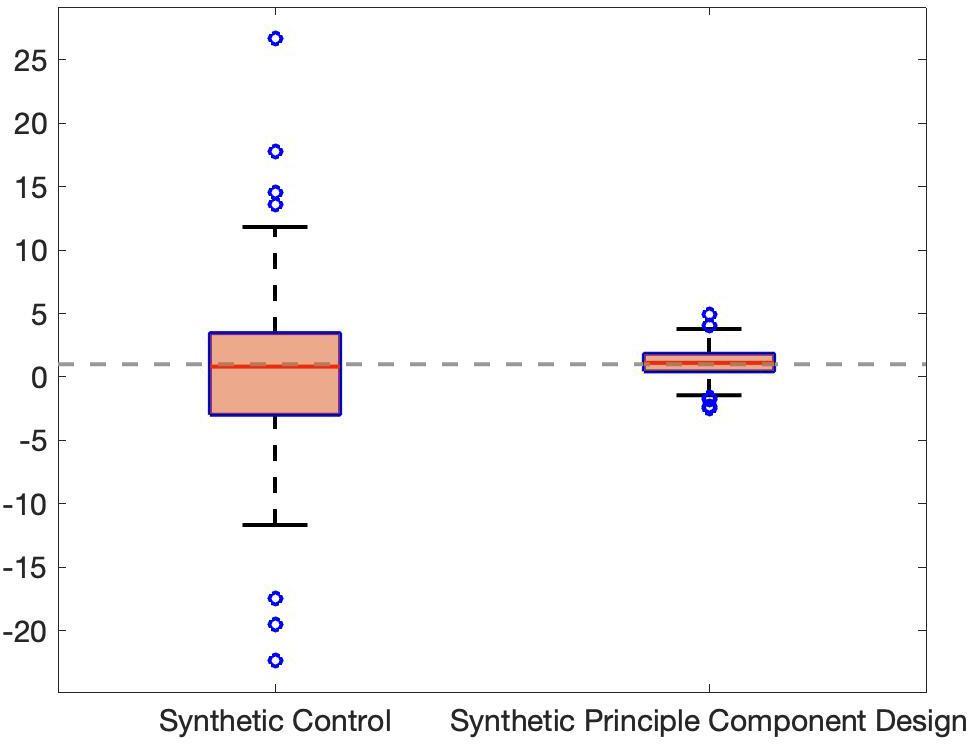}
    \label{subfig:a}}
  \quad
  \subfloat[{\scriptsize$L=20,T=20$}]{
    \includegraphics[width=0.18\textwidth]{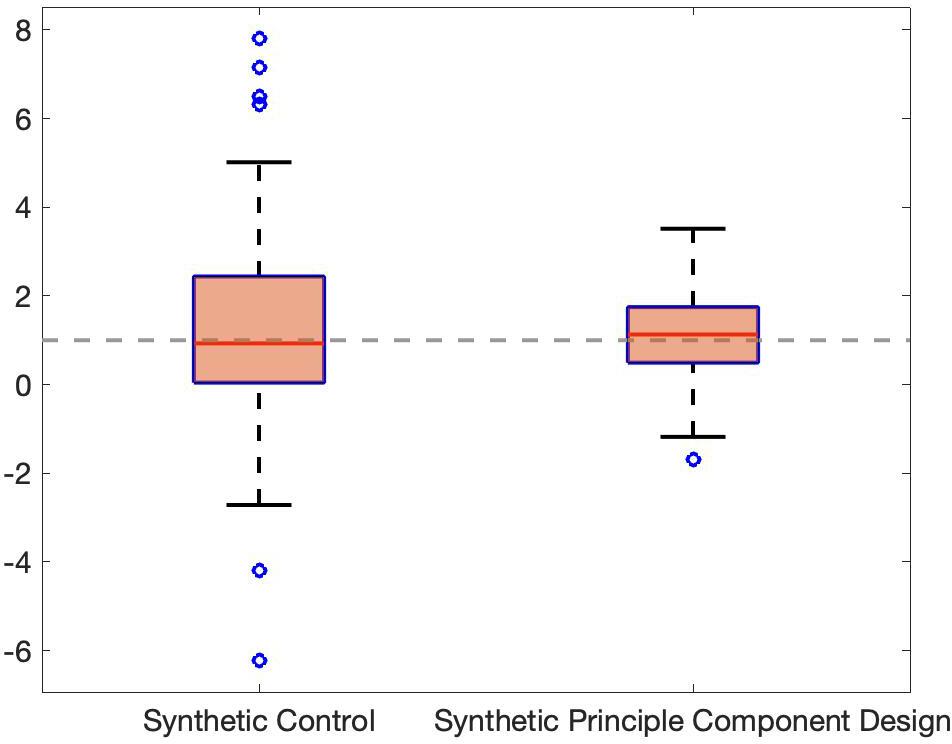}\label{subfig:b}}
  \quad
  \subfloat[{\scriptsize$L=8,T=9$}]{
    \includegraphics[width=0.18\textwidth]{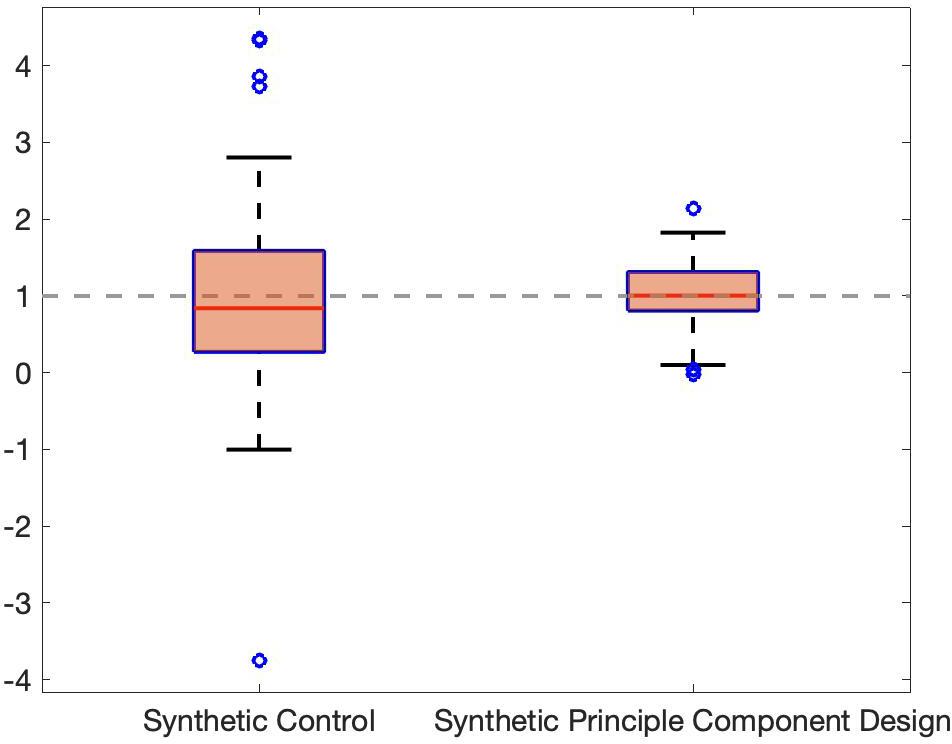}
    \label{subfig:c}}
  \quad
  \subfloat[{\scriptsize$L=8,T=20$}]{
    \includegraphics[width=0.18\textwidth]{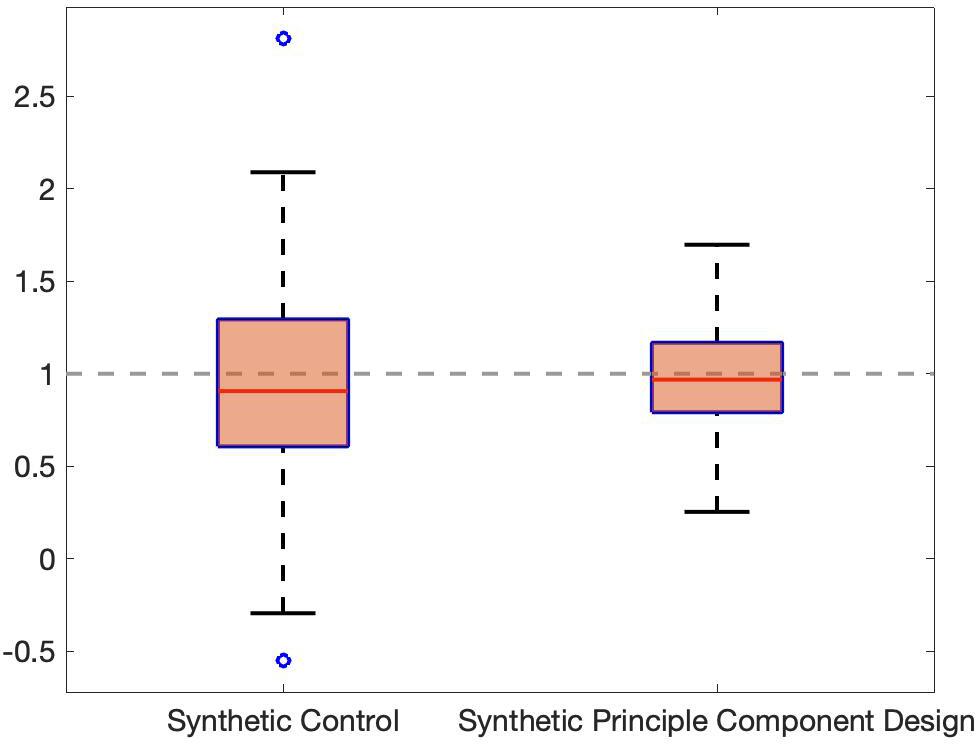}
    \label{subfig:d}}
  \caption{Treatment estimated via synthetic control (SC) and synthetic principal component design (SPCD) over 100 runs of different seeds for different selections of $L,T$ on data generated from time varying time factor.In all cases, Synthetic Principle Design provides more robust estimate of the true treatment effect 1.}
  \label{figure:simulated}

\end{figure}

\begin{figure}
  \centering
  \subfloat[{$L=20,T=9$}]{
        \includegraphics[width=0.18\textwidth]{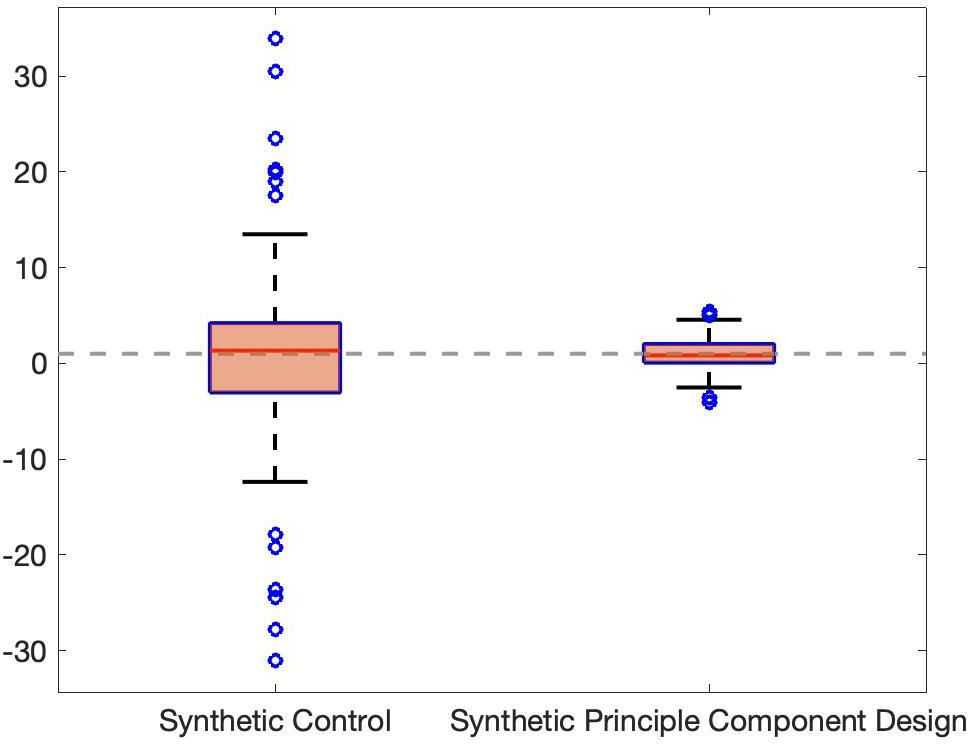}
         \label{subfig:aar}}
  \quad
  \subfloat[{$L=20,T=20$}]{
        \includegraphics[width=0.18\textwidth]{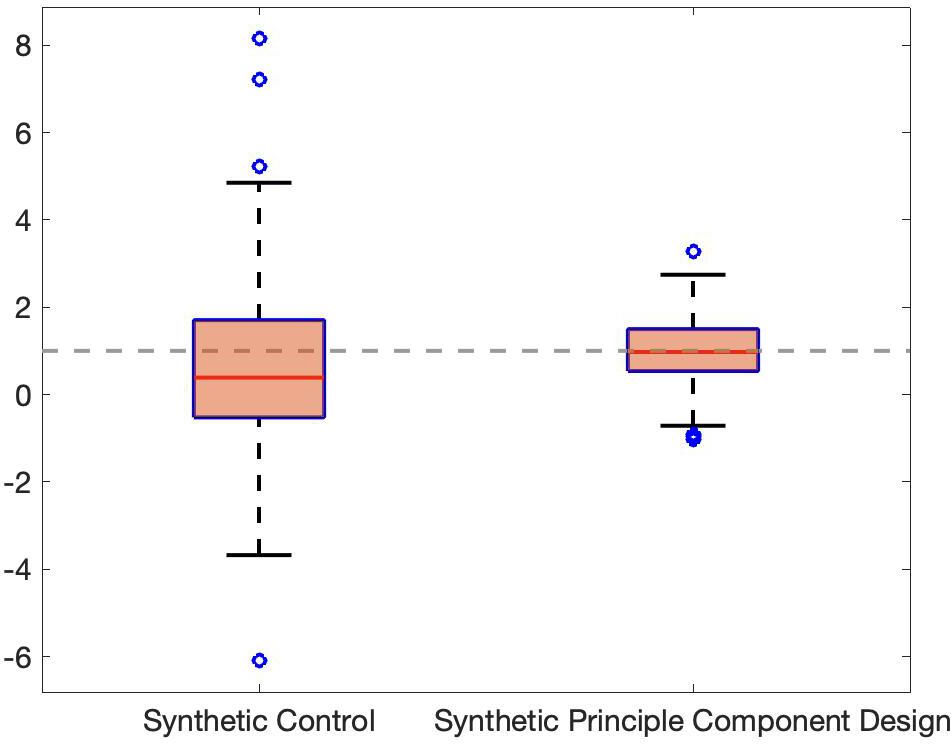}\label{subfig:bar}}
        \subfloat[{$L=8,T=9$}]{
        \includegraphics[width=0.18\textwidth]{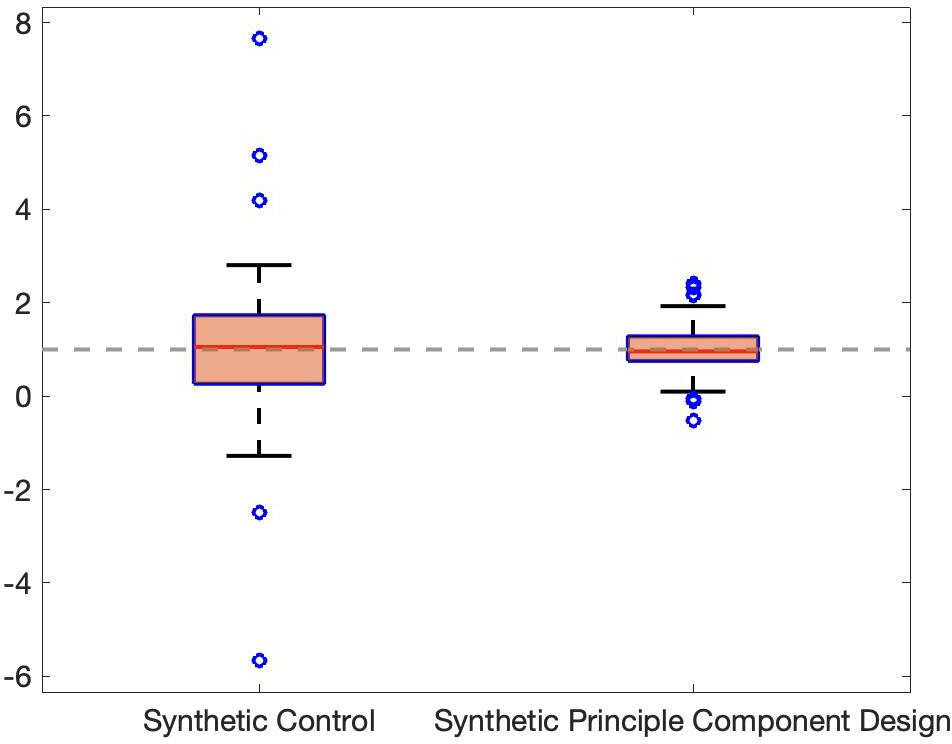}
         \label{subfig:car}}
         \quad
         \subfloat[{$L=8,T=20$}]{
        \includegraphics[width=0.18\textwidth]{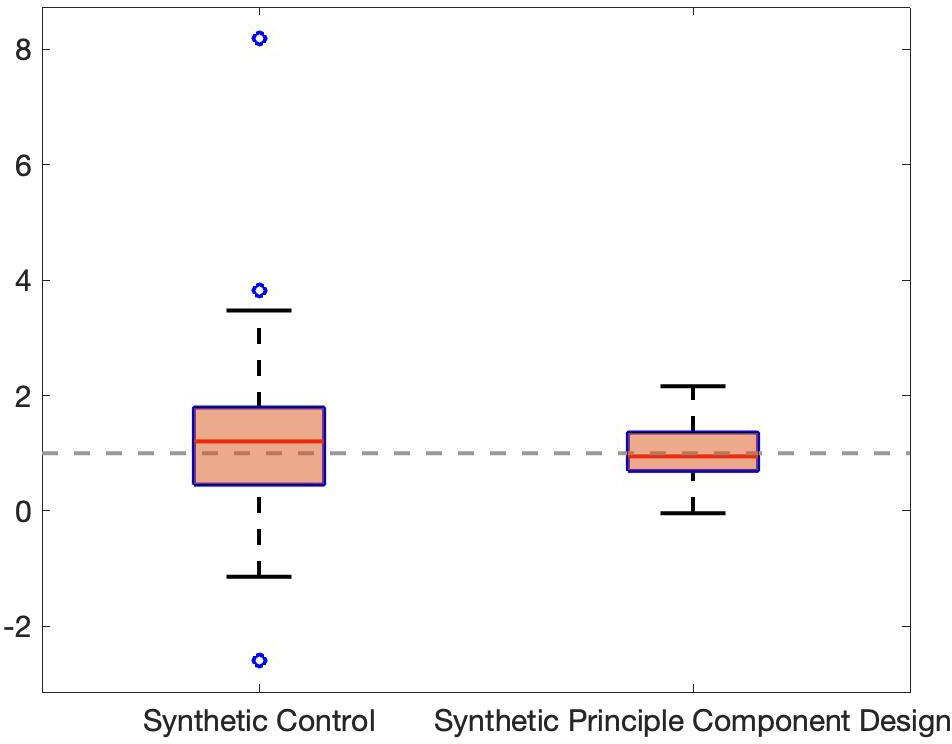}
         \label{subfig:dar}}
        \caption{Treatment estimated via Synthetic Control and Synthetic Principle Design for data generated from an AR(1) process. We run the experiment over 100 runs of different seeds for different selections of $L,T$ on data generated from purely random latent vector.In all cases, Synthetic Principle Design provides more robust estimate of the true treatment effect 1.}
        \label{figure:simulatedar}
\end{figure}

\begin{itemize}
    \item \textbf{Pure Random Latent Vector} In this experiment, we follow \citep{abadie2021synthetic} and run our algorithm on a synthetic dataset where all the time latent factors is sampled from random Gaussian. We sample the latent unit factor $v\in\mathbb{R}^{N\times T}$ and latent time factor $\gamma\in\mathbb{R}^{N\times T}$ both as random standard Gaussian matrices.  We fix the test time period $S=10$ and number of units $N=10$ and simulated different pairs of $L,T$ selection. The final results is shown in Figure \ref{figure:simulatedrandom}.
\item \textbf{Time Varying Factor} In this experiment, we generate the time factor vector $v_t$ as $t-\frac{T+S}{2}+\epsilon_t$, where
$t-\frac{T+S}{2}$ is time trend term and $e_{it}$ are i.i.d. standard Gaussian noise.  The final results are
shown in Figure \ref{figure:simulated}.
\item \textbf{AR(1) Process} In this experiment, we follow \citep{abadie2022synthetic} and run our algorithm on a synthetic dataset where the time latent factors is sampled from an AR(1) process. In particular the time factor  $\gamma=[\gamma_1,\gamma_2,\cdots,\gamma_T]'\in\mathbb{R}^{N\times T}$ is sampled via
\begin{itemize}
    \item $\gamma_1\sim\mathcal{N}(0,I_N),$
    \item $\gamma_{t+1}=A\gamma_1+b+\sigma\epsilon,\epsilon\sim\mathcal{N}(0,I_N)$.
\end{itemize}
In our experiment, we take $A=0.7I_N$, $b=\mathbbm{1}$ and $\sigma=1$. The final result is shown in Figure \ref{figure:simulatedar}.
\end{itemize}
Synthetic principal component design (SPCD) selects the treated units
whose features are representative of the whole aggregate market of interest \citep{abadie2021synthetic}. At the same time, SPCD surpasses SC in every setting in Figure \ref{figure:simulatedrandom}, \ref{figure:simulated} and \ref{figure:simulatedar}.

\vspace{-0.1in}
\subsection{Real World Data}
\label{subsection:realworld}

To exam our algorithm on real data, we follow
\citep{arkhangelsky2019synthetic,doudchenko2021synthetic} and utilize the US Bureau of Labor
Statistics and the Anti-Smoking Legislation data to examine the validity of our algorithm. Besides
synthetic control (SC), which randomly selects one unit to implement the treatment, we also
implement an additional random baseline, which randomly select units as control/group with
probability $1/2$. Through estimating the average treatment effect on the treated, we
compare SC and the random assignment baseline with SPCD in terms of the RMSE. The final result is
shown in Table \ref{table:result}. SPCD surpasses SC a large margin on both of the datasets.

\begin{figure}[h]
\centering

\subfloat[{\scriptsize  Treatment Selected when $T=25$}]{
        \includegraphics[width=0.34\textwidth]{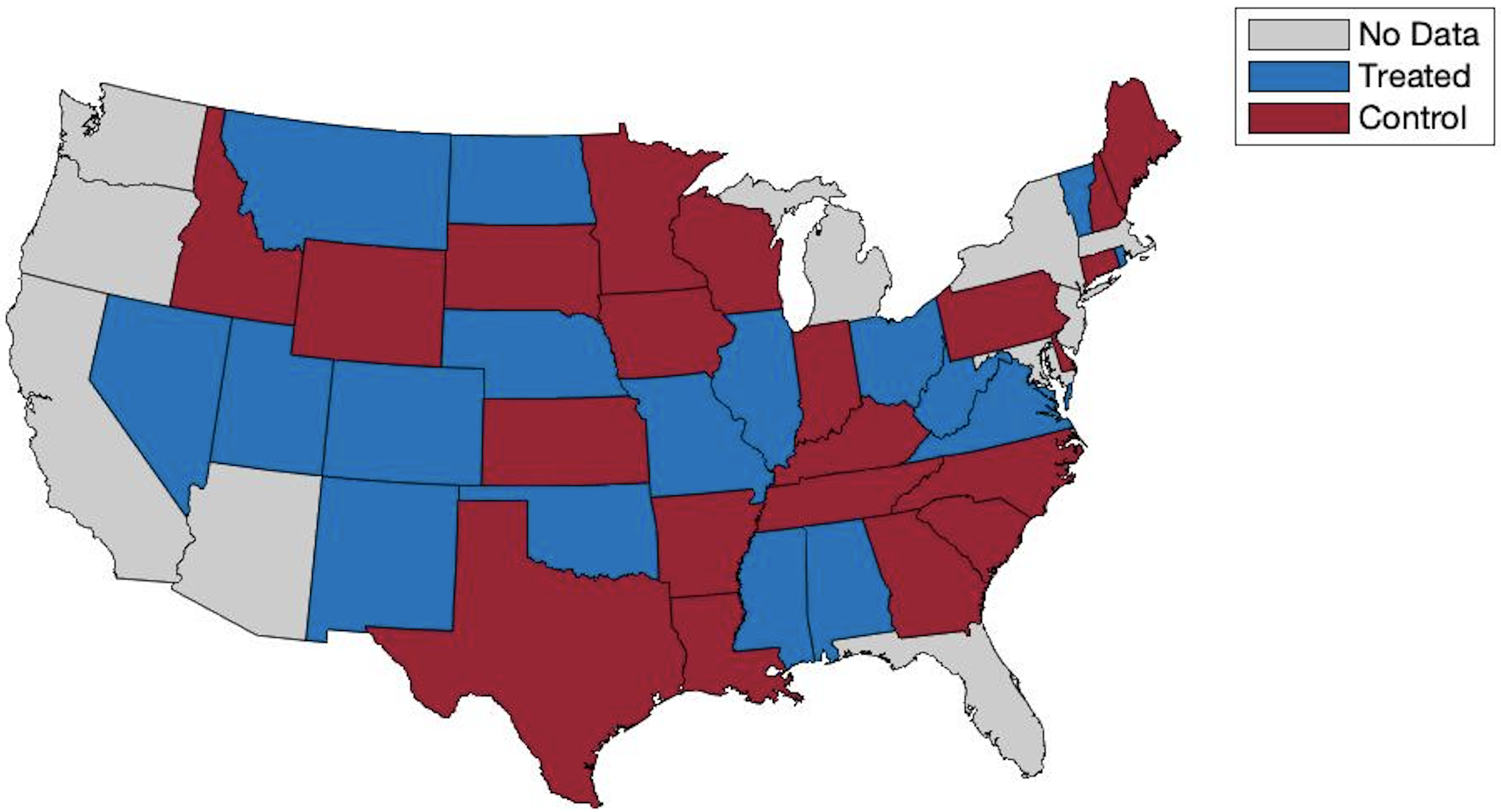}\label{subfig:usmap}}\quad
\subfloat[{\scriptsize Comparison in terms of RMSE}]{
        \includegraphics[width=0.3\textwidth]{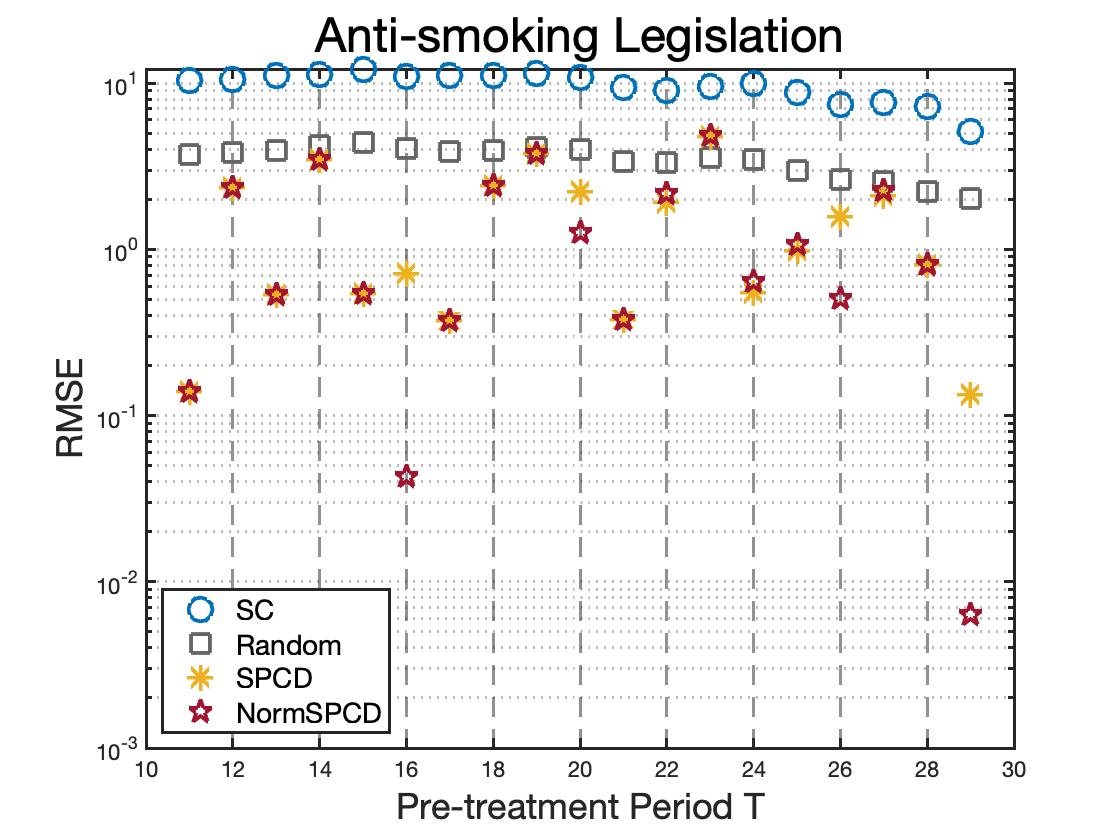}\label{subfig:differentT}}
\subfloat[{\scriptsize  Comparison with rerandomization }]{
        \includegraphics[width=0.3\textwidth]{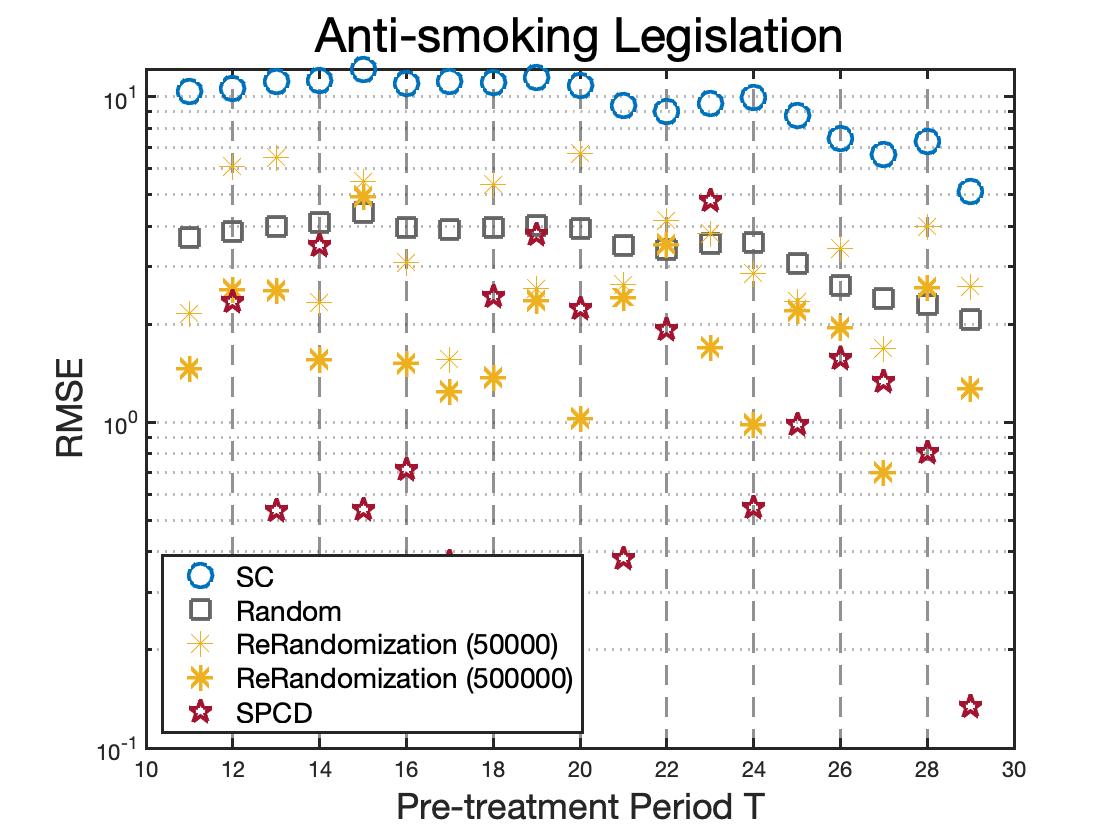}\label{subfig:rerand}}
    
  \caption{A typical design selected via synthetic principal component design (SPCD) and its performance.}
  \label{figure:treatmentUS}
\end{figure}

\textbf{The Abadie–Diamond–Hainmueller Smoking Data.} \citep{abadie2010synthetic} uses SC to study
the effects of Proposition 99, a large-scale anti-smoking legislation program that California implemented in 1988. \footnote{The Abadie–Diamond–Hainmueller Smoking data is first organized by \citep{abadie2010synthetic}. In this paper, we use the organized version by \citep{arkhangelsky2019synthetic} at \url{https://github.com/synth-inference/synthdid/blob/master/data/california_prop99.csv} which drop the data of minimum wage laws, gun laws to abortion laws in the original data and only considers the smoking outcome data.} To simulate the bias of SC and SPCD on this application, following
\citep{athey2021matrix}, we consider observations for 38 states (excluding California due to Proposition 99) from 1970 through 2000.  We regard the first $T$ year as pre-treatment periods to produce the design and use the last $31-T$ years as post-treatment periods to test the performance of the treatment assignment. The final result is shown in Table \ref{table:result} and Figure \ref{subfig:differentT}. Our design surpasses the random design by a large margin on most of the selection of time $T$. We also compare our method with the rerandomization design
\citep{morgan2012rerandomization,kallus2018optimal,li2018asymptotic} in Figure \ref{subfig:rerand}, which shows that our algorithm is still better than 500000 times of rerandomization.

One typical design produced by our algorithm is shown in Figure \ref{figure:treatmentUS}. The
experiment design for different pre-treatment length $T$ is shown in Figure
\ref{fig:differenttusamap}. The plots show that our selection of the control group is robust to
different pre-treatment time period and has the ability to represent all different geographic,
demographic, racial, and social structure of states in the United State.

\textbf{US Bureau of Labor Statistics.} We also apply our algorithm on the unemployment rate of 50 states in 40 months from the US Bureau of Labor Statistics (BLS). \footnote{{The BLS Statistics data is available
  from the BLS website. In this paper, we use the organized version by
  \citep{arkhangelsky2019synthetic} at
  \url{https://github.com/synth-inference/synthdid/blob/master/experiments/bdm/data/urate_cps.csv}. We thank \citep{arkhangelsky2019synthetic}'s authors carefully organize the data and open source it on github.} } We run 50 simulations such that each simulation utilizes a 20-by-$T+S$ matrix sampled from the original 50-by-40 dataset. More specifically, we randomly select 20 units and use the first $T$ time
period to select the synthetic design and synthetic weight. The remaining $S$ time periods are the consecutive months that follow. In our experiment, we fix $S=5$ and run both experiment for $T=5,10$. The final result in terms of the RMSE is shown in Table \ref{table:result}.

\begin{table}[]
\centering
\setlength{\tabcolsep}{6pt}
\caption{Root-mean-square errors of the average treatment effect estimates by both synthetic control (SC) and synthetic principal component design (SPCD) on real data. The random design is simulated 10 times and 95$\%$ confidence interval is demonstrated. The reported RMSE for BLS dataset are multiplied by $10^3$ for readability.}
\begin{tabular}{@{}cccccc@{}}
\hline
\hline
\toprule
\multicolumn{5}{c}{\textbf{US Bureau of Labor Statistics }}
\\
\midrule

 \multicolumn{3}{c}{$T=5$}                          & \multicolumn{3}{c}{$T=10$}                                    \\  \cmidrule(l){1-3} \cmidrule(l){4-6} 
  \multicolumn{1}{c}{SC} & \multicolumn{1}{c}{Random} & SPCD & \multicolumn{1}{c}{SC} & \multicolumn{1}{c}{Random}      & SPCD \\ \cmidrule(l){1-3} \cmidrule(l){4-6} 
                                     \multicolumn{1}{c}{14.5}                  & \multicolumn{1}{c}{7.5}       &          \textbf{0.9}     & \multicolumn{1}{c}{11.6}         & \multicolumn{1}{c}{5.6} & \textbf{0.6}      \\
                  \midrule                   \midrule
                                
\multicolumn{6}{c}{\textbf{Anti-smoking legislation}}\\
\midrule
 \multicolumn{3}{c}{\textbf{ $T=15$}}                          & \multicolumn{3}{c}{\textbf{ $T=25$}}                                    \\ \cmidrule(l){1-3} \cmidrule(l){4-6} 
  \multicolumn{1}{c}{SC} & \multicolumn{1}{c}{Random} & SPCD & \multicolumn{1}{c}{SC} & \multicolumn{1}{c}{Random}      & SPCD \\ \cmidrule(l){1-3} \cmidrule(l){4-6} 
                                     \multicolumn{1}{c}{11.65}                  & \multicolumn{1}{c}{4.32$\pm$0.21}       &           \textbf{1.14}    & \multicolumn{1}{c}{7.89}         & \multicolumn{1}{c}{3.13$\pm$0.19} & \textbf{0.98}     \\  \bottomrule
\end{tabular}
\label{table:result}
\vspace{-0.15in}
\end{table}

\begin{figure}
\centering
        \subfloat[{$T=15$}]{
        \includegraphics[width=0.25\textwidth]{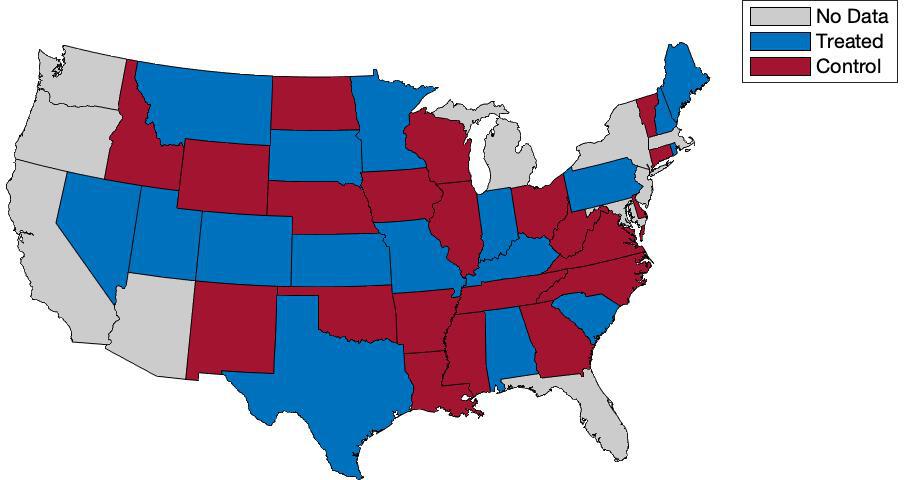}
         \label{subfig:a2}}
         \quad
         \subfloat[{$T=20$}]{
        \includegraphics[width=0.25\textwidth]{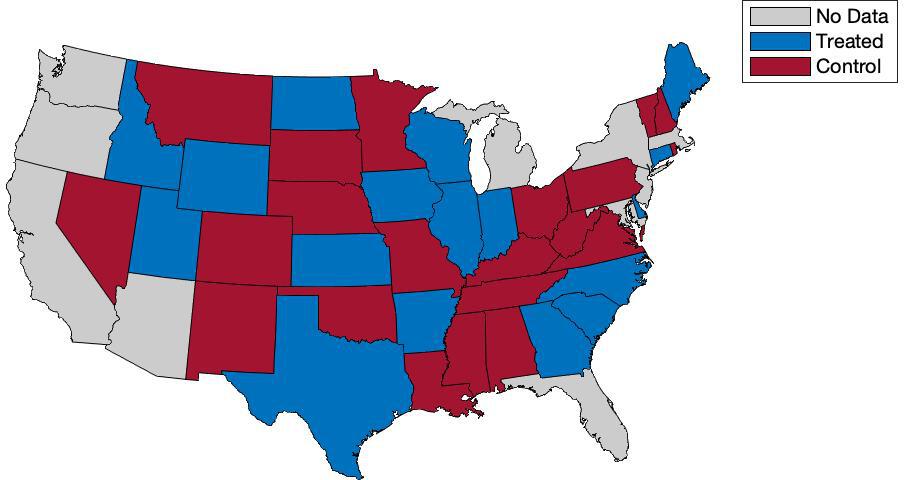}
         \label{subfig:b2}}
         \quad
         \subfloat[{$T=30$}]{
        \includegraphics[width=0.25\textwidth]{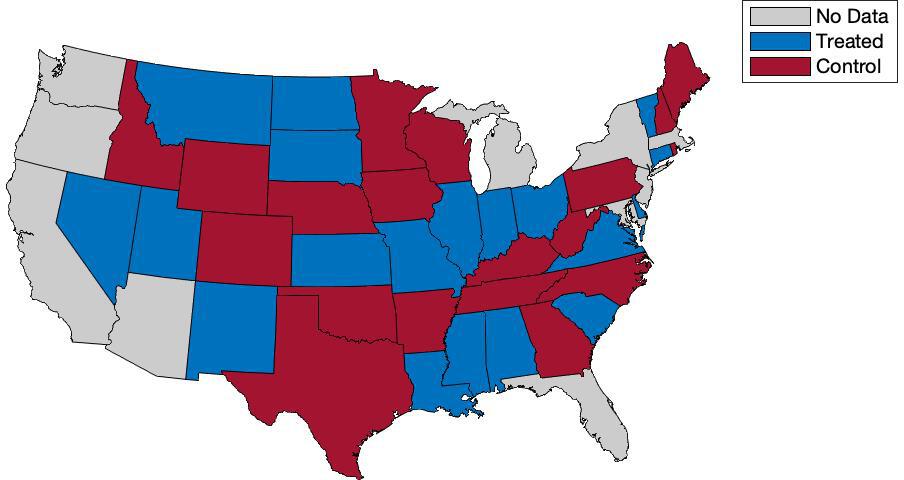}
         \label{subfig:c2}}
    \caption{Selection of control and treatment group in the Abadie–Diamond–Hainmueller California Smoking Data when different pre-treatment period length $T$ is available. The experiment design when $T=25$ is shown in Figure \ref{subfig:usmap}.}
    \label{fig:differenttusamap}
    \vspace{-0.1in}
\end{figure}

\vspace{-0.1in}
\section{Conclusion}
\label{section:conl}

In this paper, we consider the optimal experiment design problem for synthetic control, where an
NP-hard weighted covariate balancing problem is needed to be solve. Surprisingly, we reformulate the
problem as a phase synchronization problem and propose a fast spectral initialized (normalized)
generalized power method to address the resulting optimization problem efficiently. In face of a
realizable linear factor model, we provide the first global optimization results for experiment
design. Empirically, our method surpasses the original Synthetic Control and random design a large
margin in terms of RMSE on various datasets. Additionally, the newly proposed normalization
incorporated in GPM may have separate applications in degree correlated stochastic block models.

\begin{algorithm}
\caption{Empirical Implementation of SPCD}\label{alg:realalg}
\begin{algorithmic}
\Require Pre-treatment Observations $Y\in\mathbb{R}^{T\times N}$
\State Set initial treatment assignment guess through $y^0=\textnormal{sgn}(v)$, where $v$ is the smallest eigenvector of matrix $(YY^\top+\alpha I+\lambda \mathbbm{1}\mathbbm{1}^\top)$,  where $\alpha,\lambda$ are two pre-defined hyper-parameter.
\State\Comment{Spectral Initialization}
\While{Converged}
\State Select one of the following two boxes to iterate
\begin{myorangebox}
\State For SPCD, update the design via \Comment{\textbf{Generalized power methods}} 
\begin{equation}
    \begin{aligned}
    {\color{myorange}y^{t+1}=\textnormal{sgn}\left[\left((YY^\top+\alpha I+\lambda \mathbbm{1}\mathbbm{1}^\top)^{-1}+\beta I\right) y^{t}\right]},
    \end{aligned}
\end{equation}
\State where $\beta$ is a pre-defined hyper-parameter.
\end{myorangebox}
\begin{mybluebox}
\State For NormSPCD, update the design via \Comment{\textbf{Normalize the inverse covariance matrix}}
\begin{equation}
    \begin{aligned}
    {\color{myblue}y^{t+1}=\textnormal{sgn}\Big[\left[(YY^\top+\alpha I+\lambda \mathbbm{1}\mathbbm{1}^\top)^{-1}+\beta I\right] (y^{t}/d)\Big]},
    \end{aligned}
\end{equation}
\State where {$d=\sqrt{\textnormal{diag}((YY^\top+\alpha I+\lambda \mathbbm{1}\mathbbm{1}^\top)^{-1})}$} and $/$ denotes element-wise divide.
\end{mybluebox}
\EndWhile
\State Once obtained the optimal design $y^\ast$, one can select the design weight $w$ via
\begin{equation}
    \begin{aligned}
w=\frac{2(YY^\top+\alpha I+\lambda \mathbbm{1}\mathbbm{1}^\top)^{-1}y^\ast}{||(YY^\top+\alpha I+\lambda \mathbbm{1}\mathbbm{1}^\top)^{-1}y^\ast||_1}
    \end{aligned}
    \label{eq:weight}
\end{equation}
\State\Comment{The optimality condition ensures $\textnormal{sgn}(w)=y$.} 
\State Treat Unit $i$ if $y(i)=-\textnormal{sgn}\left(\sum_{i=1}^N y(i)\right)$ and run the experiment. 
\State Estimate the treatment effect via
$$
\hat\tau=\sum_{t=1}^S\left(\sum_{i=1}^N w(i)Y_{i,T+t}\right)
$$
\end{algorithmic}
\end{algorithm}

\begin{acknowledgement}
 Yiping Lu is supported by the Stanford Interdisciplinary Graduate
Fellowship (SIGF). Jose Blanchet is supported in part by the Air Force Office of Scientific Research under award number FA9550-20-1-0397 and NSF grants 1915967, 1820942, 1838576. Lexing Ying is supported by National Science Foundation under award DMS-2011699.  Yiping Lu also thanks Yitan Wang for helpful comments and feedback.
\end{acknowledgement}


\printbibliography

@article{doudchenko2021synthetic,
  title={Synthetic Design: An Optimization Approach to Experimental Design with Synthetic Controls},
  author={Doudchenko, Nick and Khosravi, Khashayar and Pouget-Abadie, Jean and Lahaie, Sebastien and Lubin, Miles and Mirrokni, Vahab and Spiess, Jann and others},
  journal={Advances in Neural Information Processing Systems},
  volume={34},
  year={2021}
}

@article{harshaw2019balancing,
  title={Balancing covariates in randomized experiments with the Gram--Schmidt Walk design},
  author={Harshaw, Christopher and S{\"a}vje, Fredrik and Spielman, Daniel and Zhang, Peng},
  journal={arXiv preprint arXiv:1911.03071},
  year={2019}
}

@inproceedings{bansal2018gram,
  title={The Gram-Schmidt walk: a cure for the Banaszczyk blues},
  author={Bansal, Nikhil and Dadush, Daniel and Garg, Shashwat and Lovett, Shachar},
  booktitle={Proceedings of the 50th Annual ACM SIGACT Symposium on Theory of Computing},
  pages={587--597},
  year={2018}
}

@article{abadie2015comparative,
  title={Comparative politics and the synthetic control method},
  author={Abadie, Alberto and Diamond, Alexis and Hainmueller, Jens},
  journal={American Journal of Political Science},
  volume={59},
  number={2},
  pages={495--510},
  year={2015},
  publisher={Wiley Online Library}
}

@article{abadie2021synthetic,
  title={Synthetic controls for experimental design},
  author={Abadie, Alberto and Zhao, Jinglong},
  journal={arXiv preprint arXiv:2108.02196},
  year={2021}
}

@article{bansal2019algorithm,
  title={An algorithm for Koml{\'o}s conjecture matching Banaszczyk's bound},
  author={Bansal, Nikhil and Dadush, Daniel and Garg, Shashwat},
  journal={SIAM Journal on Computing},
  volume={48},
  number={2},
  pages={534--553},
  year={2019},
  publisher={SIAM}
}

@article{athey2021matrix,
  title={Matrix completion methods for causal panel data models},
  author={Athey, Susan and Bayati, Mohsen and Doudchenko, Nikolay and Imbens, Guido and Khosravi, Khashayar},
  journal={Journal of the American Statistical Association},
  volume={116},
  number={536},
  pages={1716--1730},
  year={2021},
  publisher={Taylor \& Francis}
}

@article{zhong2018near,
  title={Near-optimal bounds for phase synchronization},
  author={Zhong, Yiqiao and Boumal, Nicolas},
  journal={SIAM Journal on Optimization},
  volume={28},
  number={2},
  pages={989--1016},
  year={2018},
  publisher={SIAM}
}

@article{wang2021linear,
  title={Linear Convergence of a Proximal Alternating Minimization Method with Extrapolation for $\ell_1 $-Norm Principal Component Analysis},
  author={Wang, Peng and Liu, Huikang and So, Anthony Man-Cho},
  journal={arXiv preprint arXiv:2107.07107},
  year={2021}
}

@article{mccoy2011two,
  title={Two proposals for robust PCA using semidefinite programming},
  author={McCoy, Michael and Tropp, Joel A},
  journal={Electronic Journal of Statistics},
  volume={5},
  pages={1123--1160},
  year={2011},
  publisher={Institute of Mathematical Statistics and Bernoulli Society}
}

@article{boumal2016nonconvex,
  title={Nonconvex phase synchronization},
  author={Boumal, Nicolas},
  journal={SIAM Journal on Optimization},
  volume={26},
  number={4},
  pages={2355--2377},
  year={2016},
  publisher={SIAM}
}

@article{singer2011angular,
  title={Angular synchronization by eigenvectors and semidefinite programming},
  author={Singer, Amit},
  journal={Applied and computational harmonic analysis},
  volume={30},
  number={1},
  pages={20--36},
  year={2011},
  publisher={Elsevier}
}

@article{zhang2006complex,
  title={Complex quadratic optimization and semidefinite programming},
  author={Zhang, Shuzhong and Huang, Yongwei},
  journal={SIAM Journal on Optimization},
  volume={16},
  number={3},
  pages={871--890},
  year={2006},
  publisher={SIAM}
}

@article{imai2014covariate,
  title={Covariate balancing propensity score},
  author={Imai, Kosuke and Ratkovic, Marc},
  journal={Journal of the Royal Statistical Society: Series B (Statistical Methodology)},
  volume={76},
  number={1},
  pages={243--263},
  year={2014},
  publisher={Wiley Online Library}
}

@article{zhao2019covariate,
  title={Covariate balancing propensity score by tailored loss functions},
  author={Zhao, Qingyuan},
  journal={The Annals of Statistics},
  volume={47},
  number={2},
  pages={965--993},
  year={2019},
  publisher={Institute of Mathematical Statistics}
}

@article{morgan2012rerandomization,
  title={Rerandomization to improve covariate balance in experiments},
  author={Morgan, Kari Lock and Rubin, Donald B},
  journal={The Annals of Statistics},
  volume={40},
  number={2},
  pages={1263--1282},
  year={2012},
  publisher={Institute of Mathematical Statistics}
}

@article{efron1971forcing,
  title={Forcing a sequential experiment to be balanced},
  author={Efron, Bradley},
  journal={Biometrika},
  volume={58},
  number={3},
  pages={403--417},
  year={1971},
  publisher={Oxford University Press}
}

@article{kasy2016experimenters,
  title={Why experimenters might not always want to randomize, and what they could do instead},
  author={Kasy, Maximilian},
  journal={Political Analysis},
  volume={24},
  number={3},
  pages={324--338},
  year={2016},
  publisher={Cambridge University Press}
}

@article{rubin2008objective,
  title={For objective causal inference, design trumps analysis},
  author={Rubin, Donald B},
  journal={The annals of applied statistics},
  volume={2},
  number={3},
  pages={808--840},
  year={2008},
  publisher={Institute of Mathematical Statistics}
}

@article{abadie2010synthetic,
  title={Synthetic control methods for comparative case studies: Estimating the effect of California’s tobacco control program},
  author={Abadie, Alberto and Diamond, Alexis and Hainmueller, Jens},
  journal={Journal of the American statistical Association},
  volume={105},
  number={490},
  pages={493--505},
  year={2010},
  publisher={Taylor \& Francis}
}

@article{abadie2003economic,
  title={The economic costs of conflict: A case study of the Basque Country},
  author={Abadie, Alberto and Gardeazabal, Javier},
  journal={American economic review},
  volume={93},
  number={1},
  pages={113--132},
  year={2003}
}

@article{chen2021spectral,
  title={Spectral methods for data science: A statistical perspective},
  author={Chen, Yuxin and Chi, Yuejie and Fan, Jianqing and Ma, Cong and others},
  journal={Foundations and Trends{\textregistered} in Machine Learning},
  volume={14},
  number={5},
  pages={566--806},
  year={2021},
  publisher={Now Publishers, Inc.}
}

@article{bandeira2017tightness,
  title={Tightness of the maximum likelihood semidefinite relaxation for angular synchronization},
  author={Bandeira, Afonso S and Boumal, Nicolas and Singer, Amit},
  journal={Mathematical Programming},
  volume={163},
  number={1},
  pages={145--167},
  year={2017},
  publisher={Springer}
}

@article{liu2017estimation,
  title={On the estimation performance and convergence rate of the generalized power method for phase synchronization},
  author={Liu, Huikang and Yue, Man-Chung and Man-Cho So, Anthony},
  journal={SIAM Journal on Optimization},
  volume={27},
  number={4},
  pages={2426--2446},
  year={2017},
  publisher={SIAM}
}

@article{xu2017generalized,
  title={Generalized synthetic control method: Causal inference with interactive fixed effects models},
  author={Xu, Yiqing},
  journal={Political Analysis},
  volume={25},
  number={1},
  pages={57--76},
  year={2017},
  publisher={Cambridge University Press}
}

@article{ferman2021properties,
  title={On the properties of the synthetic control estimator with many periods and many controls},
  author={Ferman, Bruno},
  journal={Journal of the American Statistical Association},
  volume={116},
  number={536},
  pages={1764--1772},
  year={2021},
  publisher={Taylor \& Francis}
}

@article{rubinstein2001reconstruction,
  title={Reconstruction of optical surfaces from ray data},
  author={Rubinstein, J and Wolansky, G},
  journal={Optical review},
  volume={8},
  number={4},
  pages={281--283},
  year={2001},
  publisher={Springer}
}

@inproceedings{giridhar2006distributed,
  title={Distributed clock synchronization over wireless networks: Algorithms and analysis},
  author={Giridhar, Arvind and Kumar, Praveen R},
  booktitle={Proceedings of the 45th IEEE Conference on Decision and Control},
  pages={4915--4920},
  year={2006},
  organization={IEEE}
}

@article{cucuringu2016sync,
  title={Sync-rank: Robust ranking, constrained ranking and rank aggregation via eigenvector and SDP synchronization},
  author={Cucuringu, Mihai},
  journal={IEEE Transactions on Network Science and Engineering},
  volume={3},
  number={1},
  pages={58--79},
  year={2016},
  publisher={IEEE}
}

@article{alexeev2014phase,
  title={Phase retrieval with polarization},
  author={Alexeev, Boris and Bandeira, Afonso S and Fickus, Matthew and Mixon, Dustin G},
  journal={SIAM Journal on Imaging Sciences},
  volume={7},
  number={1},
  pages={35--66},
  year={2014},
  publisher={SIAM}
}

@article{wang2013exact,
  title={Exact and stable recovery of rotations for robust synchronization},
  author={Wang, Lanhui and Singer, Amit},
  journal={Information and Inference: A Journal of the IMA},
  volume={2},
  number={2},
  pages={145--193},
  year={2013},
  publisher={Oxford University Press}
}

@article{singer2011three,
  title={Three-dimensional structure determination from common lines in cryo-EM by eigenvectors and semidefinite programming},
  author={Singer, Amit and Shkolnisky, Yoel},
  journal={SIAM journal on imaging sciences},
  volume={4},
  number={2},
  pages={543--572},
  year={2011},
  publisher={SIAM}
}

@inproceedings{martinec2007robust,
  title={Robust rotation and translation estimation in multiview reconstruction},
  author={Martinec, Daniel and Pajdla, Tomas},
  booktitle={2007 IEEE Conference on Computer Vision and Pattern Recognition},
  pages={1--8},
  year={2007},
  organization={IEEE}
}

@article{journee2010generalized,
  title={Generalized power method for sparse principal component analysis.},
  author={Journ{\'e}e, Michel and Nesterov, Yurii and Richt{\'a}rik, Peter and Sepulchre, Rodolphe},
  journal={Journal of Machine Learning Research},
  volume={11},
  number={2},
  year={2010}
}

@article{luss2013conditional,
  title={Conditional gradient algorithmsfor rank-one matrix approximations with a sparsity constraint},
  author={Luss, Ronny and Teboulle, Marc},
  journal={siam REVIEW},
  volume={55},
  number={1},
  pages={65--98},
  year={2013},
  publisher={SIAM}
}

@techreport{arkhangelsky2019synthetic,
  title={Synthetic difference in differences},
  author={Arkhangelsky, Dmitry and Athey, Susan and Hirshberg, David A and Imbens, Guido W and Wager, Stefan},
  year={2019},
  institution={National Bureau of Economic Research}
}

@article{shi2021assumptions,
  title={On the Assumptions of Synthetic Control Methods},
  author={Shi, Claudia and Sridhar, Dhanya and Misra, Vishal and Blei, David M},
  journal={arXiv preprint arXiv:2112.05671},
  year={2021}
}

@article{li2020statistical,
  title={Statistical inference for average treatment effects estimated by synthetic control methods},
  author={Li, Kathleen T},
  journal={Journal of the American Statistical Association},
  volume={115},
  number={532},
  pages={2068--2083},
  year={2020},
  publisher={Taylor \& Francis}
}

@article{acemoglu2016value,
  title={The value of connections in turbulent times: Evidence from the United States},
  author={Acemoglu, Daron and Johnson, Simon and Kermani, Amir and Kwak, James and Mitton, Todd},
  journal={Journal of Financial Economics},
  volume={121},
  number={2},
  pages={368--391},
  year={2016},
  publisher={Elsevier}
}

@article{cunningham2018decriminalizing,
  title={Decriminalizing indoor prostitution: Implications for sexual violence and public health},
  author={Cunningham, Scott and Shah, Manisha},
  journal={The Review of Economic Studies},
  volume={85},
  number={3},
  pages={1683--1715},
  year={2018},
  publisher={Oxford University Press}
}

@article{kleven2013taxation,
  title={Taxation and international migration of superstars: Evidence from the European football market},
  author={Kleven, Henrik Jacobsen and Landais, Camille and Saez, Emmanuel},
  journal={American economic review},
  volume={103},
  number={5},
  pages={1892--1924},
  year={2013}
}

@article{amjad2018robust,
  title={Robust synthetic control},
  author={Amjad, Muhammad and Shah, Devavrat and Shen, Dennis},
  journal={The Journal of Machine Learning Research},
  volume={19},
  number={1},
  pages={802--852},
  year={2018},
  publisher={JMLR. org}
}

@article{kwak2008principal,
  title={Principal component analysis based on L1-norm maximization},
  author={Kwak, Nojun},
  journal={IEEE transactions on pattern analysis and machine intelligence},
  volume={30},
  number={9},
  pages={1672--1680},
  year={2008},
  publisher={IEEE}
}

@article{kallus2021optimality,
  title={On the optimality of randomization in experimental design: How to randomize for minimax variance and design-based inference},
  author={Kallus, Nathan},
  journal={Journal of the Royal Statistical Society: Series B (Statistical Methodology)},
  volume={83},
  number={2},
  pages={404--409},
  year={2021},
  publisher={Wiley Online Library}
}

@article{kallus2018optimal,
  title={Optimal a priori balance in the design of controlled experiments},
  author={Kallus, Nathan},
  journal={Journal of the Royal Statistical Society: Series B (Statistical Methodology)},
  volume={80},
  number={1},
  pages={85--112},
  year={2018},
  publisher={Wiley Online Library}
}

@article{li2018asymptotic,
  title={Asymptotic theory of rerandomization in treatment--control experiments},
  author={Li, Xinran and Ding, Peng and Rubin, Donald B},
  journal={Proceedings of the National Academy of Sciences},
  volume={115},
  number={37},
  pages={9157--9162},
  year={2018},
  publisher={National Acad Sciences}
}

@article{smith1994optimization,
  title={Optimization techniques on Riemannian manifolds},
  author={Smith, Steven T},
  journal={Fields institute communications},
  volume={3},
  number={3},
  pages={113--135},
  year={1994}
}

@article{jin2015fast,
  title={Fast community detection by SCORE},
  author={Jin, Jiashun},
  journal={The Annals of Statistics},
  volume={43},
  number={1},
  pages={57--89},
  year={2015},
  publisher={Institute of Mathematical Statistics}
}

@article{zhao2012consistency,
  title={Consistency of community detection in networks under degree-corrected stochastic block models},
  author={Zhao, Yunpeng and Levina, Elizaveta and Zhu, Ji},
  journal={The Annals of Statistics},
  volume={40},
  number={4},
  pages={2266--2292},
  year={2012},
  publisher={Institute of Mathematical Statistics}
}

@article{jin2022improvements,
  title={Improvements on score, especially for weak signals},
  author={Jin, Jiashun and Ke, Zheng Tracy and Luo, Shengming},
  journal={Sankhya A},
  volume={84},
  number={1},
  pages={127--162},
  year={2022},
  publisher={Springer}
}

@article{chen2015solving,
  title={Solving random quadratic systems of equations is nearly as easy as solving linear systems},
  author={Chen, Yuxin and Candes, Emmanuel},
  journal={Advances in Neural Information Processing Systems},
  volume={28},
  year={2015}
}

@article{candes2015phase,
  title={Phase retrieval via Wirtinger flow: Theory and algorithms},
  author={Candes, Emmanuel J and Li, Xiaodong and Soltanolkotabi, Mahdi},
  journal={IEEE Transactions on Information Theory},
  volume={61},
  number={4},
  pages={1985--2007},
  year={2015},
  publisher={IEEE}
}

@article{mishra2016riemannian,
  title={Riemannian preconditioning},
  author={Mishra, Bamdev and Sepulchre, Rodolphe},
  journal={SIAM Journal on Optimization},
  volume={26},
  number={1},
  pages={635--660},
  year={2016},
  publisher={SIAM}
}

@article{wright1999numerical,
  title={Numerical optimization},
  author={Wright, Stephen and Nocedal, Jorge and others},
  journal={Springer Science},
  volume={35},
  number={67-68},
  pages={7},
  year={1999}
}

@article{abadie2022synthetic,
  title={Synthetic Controls in Action},
  author={Abadie, Alberto and Vives-i-Bastida, Jaume},
  journal={arXiv preprint arXiv:2203.06279},
  year={2022}
}

@article{tropp2015introduction,
  title={An introduction to matrix concentration inequalities},
  author={Tropp, Joel A},
  journal={arXiv preprint arXiv:1501.01571},
  year={2015}
}

@article{wang2021non,
  title={Non-convex exact community recovery in stochastic block model},
  author={Wang, Peng and Zhou, Zirui and So, Anthony Man-Cho},
  journal={Mathematical Programming},
  pages={1--37},
  year={2021},
  publisher={Springer}
}

@article{tao2011topics,
  title={Topics in random matrix theory},
  author={Tao, Terence},
  journal={Graduate Studies in Mathematics},
  volume={132},
  year={2011},
  publisher={Citeseer}
}

\appendix

\newpage
{\huge{\textbf{Appendix}}}

\section{Organization of the Appendix}
We organize the appendix as following:
\begin{itemize}
    \item In Appendix \ref{section:equal}, we demonstrated the equivalence between Synthetic Design, $\ell_1$-PCA and Phase Synchronization. We also briefly introduce the literature of solving  $\ell_1$-PCA and  Phase Synchronization in Appendix \ref{section:equal}.
    \item In Appendix \ref{section:opt}, we demonstrated the proof of global convergence of the generalized power method. The road-map of the proof is following. In Appendix \ref{subsection:Spectralinit}, we analyze the spectral initialization. We showed that it provide as accuracy estimate as the global optima to the ground truth signal. In Appendix \ref{subsection:GPW}, we verify the global optimality of the stationary point of GPW via the Riemann Hessian. In Appendix \ref{subsection:linearrate}, we demonstrated the linear convergence rate of the GPW method. In Appendix \ref{subsection:generative}, we analyze the data generating process to match the assumption needed in the global optimality of the GPW method.
\end{itemize}

\section{Equivalent to Phase Synchronization}
\label{section:equal}

\begin{theorem} For large enough $\lambda$, the global solution $W^\ast$ of (\ref{eq:min_l1}) satisfies
$$
\textnormal{sgn}(W^\ast)=\textnormal{sgn}\left(\mathop{\arg\min}_{W\in\mathbb{R}^n, ||W||_1=1 } W^\top (YY^\top+\sigma I+\lambda \mathbbm{1}\mathbbm{1}^\top) W\right)
$$
\end{theorem}

\begin{proof} We denote $W_\lambda = \arg\min_{W\in\mathbb{R}^n, ||W||_1=1 } W^\top (YY^\top+\sigma I+\lambda \mathbbm{1}\mathbbm{1}^\top) W$ for $\lambda>0$. From \citep[Theorem 17.1]{wright1999numerical}, we know that $W_\lambda\rightarrow W^\ast$ as $\lambda\uparrow\infty$. Thus there exists $\lambda^\ast$, such that for all $\lambda>\lambda^\ast$, we have
$$
||W_\lambda-W^\ast||_\infty \le \min\{|W^\ast(i)|:W^\ast(i)\not=0\}.
$$
Thus for all $\lambda>\lambda^\ast$, we have $\textnormal{sgn}(W^\ast)=\textnormal{sgn}(W_\lambda)$.
\end{proof}

\begin{remark}
In the above discussion, we consider both $\textnormal{sgn}(0)=1$ and $\textnormal{sgn}(0)=-1$ are right. The reason is that we plug in both sign selection in to the convex programming (\ref{eq:convexdesign}) can both produce the true global optimum.
\end{remark}

\begin{theorem}[Equivalence between Synthetic Design, $\ell_1$-PCA and Phase Synchronization] If $x^\ast\in\mathbb{R}^n$ is the global solution of $\min_{||x||_1=1} ||Ax||_2^2$ for some matrix $A\in\mathbb{R}^{D\times n}$ ($D>n$) and matrix $A^\top A\in\mathbb{R}^{n\times n}$ is invertible, then $y^\ast=\textnormal{sgn}(x^\ast)$ is the global solution of $\max_{y\in\{-1,+1\}^n} y^\top((A^\top A)^{-1})^\top y$.
\end{theorem}

\begin{proof} Firstly, the problem $\min_{||x||_1=1} ||Ax||_2^2$ is equivalent to  $\min_{||x||_1=1} ||(A^\top A)^{\frac{1}{2}}x||_2^2$ and can be further transformed to $\min_{||x||_2=1} ||(A^\top A)^{-\frac{1}{2}}x||_1$.

At the same time, for any matrix $T\in\mathbb{R}^{n\times n}$, we have
\begin{equation}
   \max_{||x||_2=1} ||Tx||_1 = \max_{||x||_2=1,y\in\{-1,+1\}} y^\top Tx=\max_{y\in\{-1,+1\}} ||T^\top y||_2=\max_{y\in\{-1,+1\}} y^\top TT^\top y.
   \label{eq:equalproof}
\end{equation}
and thus leads to ${\arg\max}_{y\in\{-1,+1\}} ||T^\top y||_2=\textnormal{sgn}(Tx^\ast)$ where $x^\ast={\arg\max}_{||x||_2=1} ||Tx||_1$. Combining the two facts, we can prove the theorem.
\end{proof}
\begin{remark}
Although we formulated the mixed integer programming as a well-known compact matrix form, the two problems, \emph{i.e.} $\ell_1$-PCA and phase synchronization, are still known to be NP-hard \citep{mccoy2011two,zhang2006complex}. However phase synchronization can be globally solved under certain data generative models \citep{bandeira2017tightness,boumal2016nonconvex,liu2017estimation,zhong2018near}. As far as the author known, there is still no data generative models for $\ell_1$-PCA been found can be globally solved. \citep{wang2021linear} show that for 
the Kurdyka-Lojasiewicz exponent of the $\ell_1$-PCA problem at any of the limiting critical points is $\frac{1}{2}$. This allows one to establish the linear convergence to the local stationary point of certain algorithm. Although, Generalized Power Method is also proposed for $\ell_1$-PCA \citep{kwak2008principal}, but only local convergence is guaranteed.
\end{remark}

\section{Optimization Theory}
\label{section:opt}
Out theory mostly follows \citep{bandeira2017tightness,boumal2016nonconvex}. But we have slightly different optimization problem (optimization over $\mathbb{C}^n$ in \citep{bandeira2017tightness,boumal2016nonconvex} but $\mathbb{R}^n$ in ours) and uses different data generating process (gram matrix in \citep{bandeira2017tightness,boumal2016nonconvex} and inverse of gram matrix in our paper.  All the entries of ground truth vector norm equals to 1 in \citep{bandeira2017tightness,boumal2016nonconvex}, \emph{i.e.} $|z_i|=1$. But this is not assumed in our paper.). For completeness, we complete all the proof details here in the appendix.

\subsection{Preliminaries}
In this section, we present some basics of Riemannian gradients. For $\{-1,1\}^n$ is a degenerate manifold. In the proof, we will consider the global optimality of the synchronization problem over a larger space $\mathbb{T}^n=\{z\in\mathcal{C}^n:|z_1|=\cdots=|z_n|=1\}$. Next we endow $\mathbb{T}^n$ with Euclidean metric $\left<y^1,y^2\right>=\sum_{i=1}^n\mathcal{R}\{y^1_i{y^{2}_i}^H\}$ which is the equivalent to viewing $\mathcal{C}^n$ as $\mathbb{R}^{2n}$ and equip with the canonical inner product. Then $\mathbb{T}^n$ can be considered as a sub-manifold and the tangent space can be written as
$$
\mathcal{T}_y\mathbb{T}^n=\{\dot y\in\mathcal{C}^n:\mathcal{R}(\dot y_i y_i^H)=0, \forall i \in [n]\}.
$$
The projector to the tangent space is $\textnormal{Proj}_x:\mathbb{C}^n\rightarrow T_x \mathbb{T}^n:\dot x\rightarrow \dot x- \mathcal{R}\{\textnormal{ddiag}(\dot xx^H)\}x$, where $\textnormal{ddiag}:\mathbb{C}^n\rightarrow\mathbb{C}^n$ is a function set all off-diagonal entries of the input matrix to zero. Thus the Riemannian gradient of function $f(x)=x^HCx$ is given as
$$
\textnormal{grad} f(x)=2(\mathcal{R}\{\textnormal{ddiag}(Cxx^H)\}-C)x
$$
Following \citep{bandeira2017tightness}, we consider the Riemannian Hessian on the tangent space as the second-order necessary optimality condition. The Riemannian Hessian is defined as
$$
\textnormal{Hess} f(x)[\dot x]=\textnormal{Proj}_x \textnormal{D} \textnormal{grad} g(x)[\dot x]=\textnormal{Proj}_x 2S(x)\dot x,
$$
where $S(x)=\mathcal{R}\{\textnormal{ddiag}(Cxx^H)\}-C$. If $x$ is a (local) optimum, then $\left<\dot x,S(x)\dot x\right>>0$ for all $\dot x\in \mathcal{T}_x\mathbb{T}^N$.

For NormSPCD iteration \ref{update:normsgd}, we consider the update as a Riemannian steepest-descent\citep{mishra2016riemannian}. The Riemannian steepest-descent \emph{search} direction to minimize objective function $f$ as $\arg\min_{\xi_x\in\mathbb{R}^n}\left<\nabla f(x),\xi_x\right>_\mathcal{R}+\frac{1}{2} \left<\xi_x,\xi_x\right>_\mathcal{R}$. The Riemannian metric we consider for NormSPCD on $\mathcal{T}_x\mathbb{T}^N$ defined as 
$$
\left<y_1,y_2\right>_\mathcal{R}=\sum_{i=1}^N |z_i|^2\mathcal{R}\left(y_1(i)y_2(i)^H\right)\quad ,\forall y_1,y_2\in\mathcal{T}_x\mathbb{T}^N.
$$
Similarly,  we can define the new Riemannian Hessian as $S_\mathcal{R}(x)=\mathcal{R}\{\textnormal{ddiag}(\mathring{C}xx^H)\}-\mathring{C}$, where $\mathring{C}=\textnormal{diag}(\frac{1}{{|z|}})C\textnormal{diag}(\frac{1}{{|z|}})$. We'll show that $rS(x) \preccurlyeq S_\mathcal{R}(x)\preccurlyeq RS(x)$ for some constant $r,R>0$

 {In our discussion, we consider our algorithm works in the Field of complexity numbers. However, from the closeness of the Field of real numbers, we know that the whole trajectory of our algorithm lies in the Field of real numbers. Global minimum in the complex domain is a harder problem and directly indicate the global optimality in $\{-1,+1\}^N$.}

\subsection{Global Optimality of (Normalzied) Generalized Power Methods}
 In this section, we first study a meta version of the optimization problem. Then we will show how our generative model can be fitted into this framework. We consider the following meta optimization problem
\begin{equation}
    \begin{aligned}
    \min_{x\in\mathbb{T}^n} f(x)=x^HCx
    \end{aligned}
    \label{problem:meta}
\end{equation}
where $C=zz^H+\Delta$ is a Hermite perturbed rank-1 matrix. Different from \citep{bandeira2017tightness,boumal2016nonconvex,liu2017estimation,zhong2018near} which assumes $z\in\mathbb{T}^N$, instead, we have the following assumption on the ground truth vector $z\in\mathbb{R}^N$: 
\begin{assumption} For some $\epsilon>1$, we have
$$
{\epsilon}\le|z_i|\le\frac{1}{\epsilon}, \qquad \forall i \in [N]
$$
\end{assumption}
This is a smooth optimization problem over a smooth Riemannian manifold $\mathcal{T}^n$. Then the Riemannian gradient $\textnormal{grad} g(x)=2(\mathcal{R}(\textnormal{diag}(Cxx^H))-C)x$. The first order necessary optimality condition is $\textnormal{gard}g(x)=0$. We will also make use the second order optimality via the Riemannian Hessian
$$
\textnormal{Hess} g(x)[\dot x] = 2\left<\dot x, S\dot x\right>,\forall \dot x \in \mathcal{T}_y\mathbb{T}^n
$$
where $S(x)=\mathcal{R}(\textnormal{ddiag}(Cxx^H)-C)$ and $\textnormal{ddiag}:\mathbb{C}^{n\times n}\rightarrow\mathbb{C}^{n\times n}$ zeros out all off-diagonal entries of a matrix. Although computing the global optimum of (\ref{problem:meta}) is NP-hard \citep{zhang2006complex}, fortunately, global optimality of (\ref{problem:meta}) can sometimes be certified through the Hermitian Hessian matrix $S(x)=\mathcal{R}(\textnormal{ddiag}(Cxx^H)-C)$. This can be shown in the following lemma for sufficient optimality condition:

\begin{lemma}[Optimality Gap] Let $x^\ast$ be globally optimal for (\ref{problem:meta}). For any $x\in\mathbb{T}^N$, the optimality gap at $x$ can be bounded as
$$
0\le f(x^\ast)-f(x)\le -N\lambda_{\min}(S(x)).
$$
As a result, if $S(x)\succeq 0$, then $x$ is the global optimality problem for (\ref{problem:meta}).
\label{lemma:optimalgap}
\end{lemma}
\begin{proof}
See \citep[Section 4.2]{bandeira2017tightness} and \citep[Lemma 2]{boumal2016nonconvex}.

\end{proof}

In the following lemma, we showed that similar property holds for the changed Riemannian metric.

\begin{lemma}[Optimality Gap for Riemannian Formulation] Let $x^\ast$ be globally optimal for (\ref{problem:meta}). If $\epsilon<|z_i|<\frac{1}{\epsilon}$, then for any $x\in\mathbb{T}^N$, the optimality gap at $x$ can be bounded as
$$
0\le f(x^\ast)-f(x)\le -\frac{1}{\epsilon}N\lambda_{\min}(S_\mathcal{R}(x)).
$$
As a result, if $S(x)\succeq 0$, then $x$ is the global optimality problem for (\ref{problem:meta}).
\label{lemma:optimalgapnorm}

\end{lemma}
\begin{proof} This is because
$$
x^HCx-y^HCy=y^HS(x)y\ge \frac{N}{\epsilon}\lambda_{\min}(S_\mathcal{R}(x))
$$
\vspace{-0.2in}
\end{proof}

To solve this problem, we consider the following Generalized Power Method \citep{boumal2016nonconvex,liu2017estimation} in Algorithm \ref{alg:gpmmeta} and our normalized version in Algorithm \ref{alg:ngpmmeta}.
\begin{algorithm}
\caption{Generalized Power Method}\label{alg:gpmmeta}
\begin{algorithmic}
\State Set initialization through $x^0=\textnormal{sgn}(v)$, where $v$ is the leading eigenvector of matrix $C$. \Comment{Spectral Initialization}
\State Define $\Tilde{C}=C+\alpha I_N$ where $\alpha=||\Delta||$
\While{Converged}
\State
$x^{t+1}=T(x^t)\triangleq\textnormal{sgn}\left[\Tilde{C} x^{t}\right]$ \Comment{Generalized power methods} 
\EndWhile
\end{algorithmic}
\end{algorithm}

\begin{algorithm}
\caption{Normalized Generalized Power Method}\label{alg:ngpmmeta}
\begin{algorithmic}
\State Set initialization through $x^0=\textnormal{sgn}(v)$, where $v$ is the leading eigenvector of matrix $C$. \Comment{Spectral Initialization}
\State Define $\Tilde{C}=C+\alpha I_N$ where $\alpha=||\Delta||$
\While{Converged}
\State
$x^{t+1}=\mathring{T}(x^t)\triangleq\textnormal{sgn}\left[\Tilde{C}\large( x^{t}./\sqrt{\textnormal{\small ddiag}(C)}\large)\right],$ \Comment{Normalized Generalized power methods}
\State where $./$ is the element-wise division.
\EndWhile
\end{algorithmic}
\end{algorithm}

\subsubsection{The Spectral Initialization}
\label{subsection:Spectralinit}

We make the initial guess in Algorithm \ref{alg:gpmmeta} via spectral relaxation \citep{singer2011angular}. Denote $v$ as the leading eigenvector of matrix $C$. From the following Lemma \ref{lemma1}, we can know that the leading eigenvector $v$ is close to the ground truth signal $z$.

\begin{lemma} Given vector $z\in\mathbb{R}^N$ satisfies $||z||_2=1$. For matrix $C=zz^\top+\Delta$, where $\Delta\in\mathbb{R}^{N\times N}$ is a symmetric perturbation matrix. Then for all $x\in\mathcal{C}^n$ and $||x||_2^2=N$ satisfies $x^HCx\ge z^HCz$, we have
$$
\left|\left|\min_{\theta\in\{1,-1\}}\theta x-z\right|\right|\le \frac{4||\Delta||}{\sqrt{N}}
$$
where the $||\Delta||$ is matrix operator norm.
\label{lemma1}
\end{lemma}
\begin{proof}
See \citep[Lemma 4.1]{bandeira2017tightness} and\citep[Lemma 1]{boumal2016nonconvex}.
\vspace{-0.2in}
\end{proof}

Based on the top eigenvector $v$, we project $v$ to the Riemann manifold $\mathbb{T}^n$ and make the initial guess. For $C$ is a symmetric real matrix, the eigenvector $v\in\mathbb{R}^N$. Thus projection to  $\mathbb{T}^n$ of vector $v$ will simply become $\textnormal{sgn}(v)$. In the next lemma, we'll show the Spectral Estimator is almost as close to $z$ as the global optima.

\begin{lemma} [The Spectral Estimator is almost as accuracy as the Global Optima] 
$$
\left|\left|\min_{\theta\in\{1,-1\}}\theta \textnormal{sgn}(v)-\theta \textnormal{sgn}(z)\right|\right|\le \frac{8||\Delta||}{\epsilon\sqrt{N}}
$$
\label{lemma:spectral}
\vspace{-0.2in}
\end{lemma}
Lemma \ref{lemma:spectral} is the direct corollary of Lemma \ref{lemma1} and the following technical lemma, which is also important in the convergence rate proof in Section \ref{subsection:linearrate}.
\begin{lemma}For $w\in\mathbb{R}^n$ and $z\in\mathbb{R}^n$ satisfies $||z||_2^2=N$ and $\epsilon\le|z_i|\le \frac{1}{\epsilon}, \forall i\in [N]$ (or $1+\epsilon$), then we have
$$
||\textnormal{sgn}(w)-\textnormal{sgn}(z)||_2\le \frac{2}{\epsilon}||w-z||_2.
$$
\vspace{-0.2in}
\label{lemma:sgn}
\end{lemma}
\begin{proof}
\citep[Lemma 13]{zhong2018near} and \citep[Lemma 3]{wang2021non}
\end{proof}

\subsubsection{The Generalized Power Method}
\label{subsection:GPW}

Although, the Spectral Estimator produce good estimates. We still cannot obtain the global optimum of (\ref{problem:meta}). Following \citep{boumal2016nonconvex,liu2017estimation}, we proceed the Generalized Power Method (GPM) to further improve the estimate. \citep{boumal2016nonconvex} showed that the Generalized Power Method will converge to the global optima of problem (\ref{problem:meta}) and \citep{liu2017estimation} showed that the proceeded estimate is always better than the initial spectral estimate. The procedure of the Generalized Power Method is shown in Algorithm \ref{alg:gpmmeta}. We also consider the Normalized GPM (Algorithm \ref{alg:ngpmmeta}) in this section.

For the simplicity of description, we define an equivalence relationship $\sim$ over  $\mathbb{T}^n$ as
$$
x\sim y\quad \iff \quad x=ye^{i\theta} \quad \textnormal{for some }\theta\in\mathbb{R}.
$$
The quotient space $\mathbb{T}^n/\sim$ is defined as all the corresponding equivalence class $\{xe^{i\theta}:\theta\in\mathbb{R}\}$ for some $x\in\mathbb{C}$. The error measure we are interested in 
$$
d_q(z,x)=\min_{\theta\in\mathbb{R}}||xe^{i\theta}-z||_q=\sqrt{2(n-|z^H x|)}, \quad q\in[1,\infty].
$$

\begin{lemma}
For all $x,y\in\{-1,1\}^N$ and $q\in[1,\infty]$, then we have
$$
e^{i\mathop{\arg\min}\limits_{\theta\in\mathbb{R}}||xe^{i\theta}-z||_q}\in \{-1,1\}.
$$
\end{lemma}
\begin{proof} We use proof by contradiction to prove this statement. If $\theta^*=\mathop{\arg\min}\limits_{\theta\in\mathbb{R}} e^{i||xe^{i\theta}-z||_q}\not\in \mathbb{R}$, then we will have
$$
||x\textnormal{sgn}(\mathcal{R}(e^{i\theta^*}))-z||_q\le  ||x(\mathcal{R}(e^{i\theta^*}))-z||_q<e^{i||xe^{i\theta}-z||_q}.
$$
This is contradicted with $\theta^*=\arg\min_{\theta\in\mathbb{R}} e^{i||xe^{i\theta}-z||_q}$. Thus $e^{i\arg\min_{\theta\in\mathbb{R}}||xe^{i\theta}-z||_q}\in \{-1,1\}.$

\end{proof}
Notice that (normalized) GPM iterates on the quotient space $\mathcal{T}^N/\sim$, \emph{i.e.} if $x\sim y$, then $\textnormal{sgn}(\Tilde{C}x)\sim \textnormal{sgn}(\Tilde{C}y)$. Thus without further notice, all the equality in the following discussion is equality in the quotient, \emph{i.e.} $x=y$ means $x\sim y$.

\begin{lemma}[Monotonic Cost Improvement for GPM] The iterates $\{x^k\}_{k\in\mathbb{N}}$ produced by Algorithm \ref{alg:gpmmeta} satisfies $f(x^{k+1})>f(x^k)$ unless converged. Thus the iterates $x^k$ do not cycle. 
\label{lemma:monotonicGPM}
\end{lemma}
\begin{proof}
See \citep[Lemma 8]{boumal2016nonconvex}
\end{proof}

Although Algorithm \ref{alg:ngpmmeta} does not guarantee the Monotonic Cost Improvement on the original target function $f$. We can still prove that the produced iterates from Algorithm \ref{alg:ngpmmeta} do not cycle for it's monotonically improve another energy function.

\begin{lemma}[Converging of Normalized GPM] The iterates $\{x^k\}_{k\in\mathbb{N}}$ produced by Algorithm \ref{alg:ngpmmeta} do not cycle.
\label{lemma:ngpmconverge}
\end{lemma}
\begin{proof} Consider the potential function $\mathring{f}(x)=\frac{1}{2}x^H\left(\textnormal{diag}(\frac{1}{|z|})C\textnormal{diag}(\frac{1}{|z|})\right)x$. Then the normalized power iteration can be considered as the frank-wolf algorithm for the potential function in the sense that
$$
x^{t+1}=\mathring{T}(x^t)=\arg\max_{y\in\mathbb{T}^N}\left<y,\left(\textnormal{diag}(\frac{1}{|z|})C\textnormal{diag}(\frac{1}{|z|})\right)x^t\right>.
$$
Similar with \citep[Lemma 8]{boumal2016nonconvex}, we knows that the iterates $\{x^k\}_{k\in\mathbb{N}}$ monotonically improve the potential function $\mathring{f}$ and thus the iterates do not cycle.
\end{proof}

\begin{lemma} If $x$ is a fixed point of Generalized Power Methods (Algorithm \ref{alg:gpmmeta}), at least one of the following holds
$$
|z^H x|\ge \epsilon N- 4(||\Delta||+\alpha) \qquad \textnormal{or}\qquad |z^H x|\le \frac{4(||\Delta||+\alpha)}{\epsilon}.
$$
Furthermore, if $||\Delta||\le\frac{\epsilon^2N}{13}$  and $\alpha<||\Delta||$, all the accumulation points $x$ of Algorithm \ref{alg:gpmmeta} satisfies $|z^Hx|\ge \epsilon n-8||\Delta||$.
\label{lemma:goodstationarypointgpm}
\end{lemma}
\begin{proof} The fixed point of the generalized power method satisfies $(\Tilde{C}x)_i\bar x_i=|(\Tilde{C}x)_i|$. Thus we have
\[
x^H\Tilde{C}x=||\Tilde{C}x||_1
\]
On one hand, the quadratic term $x^H\Tilde{C}x$ can be upper bounded as
\[
x^H\Tilde{C}x=|z^H x|^2+x^H\Delta x+\alpha n\le|z^H x|^2+(||\Delta||+\alpha) n.
\]
On the other hand $||\Tilde{C}x||_1$ can be lower bounded via
\[
||\Tilde{C}x||_1=\sum_{i=1}^N\left|(z^Hx)z_i+(\Delta x)_i+\alpha x_i\right|\ge N\epsilon |z^Hx|-||\Delta x||_1-\alpha N.
\]
At the same time, $||\Delta x||_1\le\sqrt{N}||\Delta x||_2\le N||\Delta||$. Combine this with the two previous inequality, we get
\[
|z^Hx|(\epsilon N-|z^Hx|)\le 2N(||\Delta||+\alpha).
\]

The above inequality enforces that one of  $
|z^H x|\ge \epsilon N- 4(||\Delta||+\alpha) $ and $|z^H x|\le \frac{4(||\Delta||+\alpha)}{\epsilon}.
$ holds. We call all the stationary point satisfies $
|z^H x|\ge \epsilon N- 4(||\Delta||+\alpha) $ "good" stationary point and the stationary point satisfies $|z^H x|\le \frac{4(||\Delta||+\alpha)}{\epsilon}$ "bad" stationary point. In the following discussion, we use Lemma \ref{lemma:spectral} to show that the spectral initialization $\textnormal{sgn}(v)$ outperforms all the "bad" fixed points. Due to Lemma \ref{lemma:monotonicGPM}, Generalized Power Method consistently improve the cost function and thus only converge to "good" stationary points. From Lemma \ref{lemma:spectral}, we have
\begin{equation}
    \begin{aligned}
    \textnormal{sgn}(v)^HC\textnormal{sgn}(v)&=|\textnormal{sgn}(v)^Hz|^2+\textnormal{sgn}(v)^H\Delta \textnormal{sgn}(v)\\
    &\ge \left(\epsilon N-\frac{32||\Delta||^2}{\epsilon^3 N}\right)^2-N||\Delta||\ge \epsilon^2N^2-\frac{64||\Delta||^2}{\epsilon^3}-N||\Delta||
    \end{aligned}
    \label{good}
\end{equation}

The first inequality is because of Lemma \ref{lemma:spectral}, which proved the estimation $\text{sgn}(v)^H\text{sgn}(z)\ge N-\frac{32||\Delta||^2}{\epsilon^2N}$ and leads to the following results $\text{sgn}(v)^Hz\ge \epsilon N-\frac{32||\Delta||^2}{\epsilon^3N}$. At the same time, all the bad fixed points $x$ satisfies
\begin{equation}
\begin{aligned}
x^HCx=|x^Hz|^2+x^H\Delta x\le \frac{64||\Delta||^2}{\epsilon^2}+N||\Delta||
\end{aligned}
\label{bad}    
\end{equation}
Combine (\ref{good}), (\ref{bad}) with the assumption $||\Delta||\le\frac{\epsilon^{3/2} N}{13}$  and $\alpha<||\Delta||$, we knows that the spectral initialization surpasses all the bad local points. 
\end{proof}

For Normalized Generalized Power Method, we can prove a similar version.

\begin{lemma} If $x$ is a fixed point of Normalized Generalized Power Methods (Algorithm \ref{alg:ngpmmeta}), at least one of the following holds
$$
|z^H x|\ge N- \frac{4}{\epsilon}(||\Delta||+\alpha) \qquad \textnormal{or}\qquad |z^H x|\le \frac{4}{\epsilon}(||\Delta||+\alpha).
$$
Furthermore, if $||\Delta||\le\frac{n\epsilon}{13}$ and $\alpha<||\Delta||$, all the accumulation points $x$ of Algorithm \ref{alg:gpmmeta} satisfies $|z^Hx|\ge n-8||\Delta||$.
\label{lemma:goodstationarypointngpm}
\end{lemma}
\begin{proof} If $x$ is a fixed point of Algorithm \ref{alg:ngpmmeta}, then $x$ satisfies $||\mathring{C} x||_1=x^H\mathring{C} x$ (because $|(\mathring{C} x)_i|=\left<x_i,(\mathring{C}x)_i\right>$ holds for all $i$) where $\mathring{C}=\textnormal{sgn}(z)\textnormal{sgn}(z)^H+\textnormal{diag}(\frac{1}{{|z|}})\Delta\textnormal{diag}(\frac{1}{{|z|}})+\alpha \textnormal{diag}(\frac{1}{|z|^2})$. For simplicity, we denote $\mathring{\Delta}=\textnormal{diag}(\frac{1}{{|z|}})\Delta\textnormal{diag}(\frac{1}{{|z|}})$ in the following proof.

On one hand 
$$
x^H\mathring{C}x=|\textnormal{sgn}(z)^Hx|^2+x^H \mathring{\Delta} x +\alpha \sum_{i=1}^N \frac{1}{|z_i|}\le |\textnormal{sgn}(z)^Hx|^2+\frac{N}{\epsilon}(||\Delta||+\alpha)
$$

On the other hand 
$$
||\mathring{C} x||_1=\sum_{i=1}^N |(\textnormal{sgn}(z)^Hx)\textnormal{sgn}(z_i)+(\mathring{\Delta}x)_i+\alpha x./|z||\ge N|\textnormal{sgn}(z)^Hx|-||\mathring{\Delta}x||_1-\frac{\alpha N}{\epsilon}
$$

At the same time, $||\mathring{\Delta} x||_1\le\sqrt{N}||\mathring{\Delta} x||_2\le N||\mathring{\Delta}||$. Combine this with the two previous inequality, we get
$$
|z^Hx|( N-|z^Hx|)\le N||\mathring{\Delta}||+ \frac{N}{\epsilon}(||\Delta||+2\alpha)\le \frac{2N}{\epsilon}(||\Delta||+\alpha).
$$

The above inequality enforces that one of  $
|z^H x|\ge N- \frac{4}{\epsilon}(||\Delta||+\alpha)$ and $|z^H x|\le \frac{4}{\epsilon}(||\Delta||+\alpha)$ holds. We call all the stationary point satisfies $|z^H x|\ge N- \frac{4}{\epsilon}(||\Delta||+\alpha)$ good stationary point and the stationary point satisfies $|z^H x|\le \frac{4}{\epsilon}(||\Delta||+\alpha)$ bad stationary point. In the following discussion, we use Lemma \ref{lemma:spectral} to show that the spectral initialization $\textnormal{sgn}(v)$ outperforms all the bad fixed points in terms of the potential function $\mathring{f}(x)=\frac{1}{2}x^H\left(\textnormal{diag}(\frac{1}{|z|})C\textnormal{diag}(\frac{1}{|z|})\right)x$. From Lemma \ref{lemma:spectral}, we have
\begin{equation}
    \begin{aligned}
    \mathring{f}(\textnormal{sgn}(v))&=|\textnormal{sgn}(v)^Hz|^2+\textnormal{sgn}(v)^H\mathring{\Delta} \textnormal{sgn}(v)\\
    &\ge \left(N-\frac{32||\Delta||^2}{\epsilon^2 N}\right)^2-\frac{N}{\epsilon}||\Delta||\ge N^2-\frac{64||\Delta||^2}{\epsilon^2}-\frac{N}{\epsilon}||\Delta||
    \end{aligned}
    \label{good1}
\end{equation}
The second equality is because
$\frac{64||\Delta||^2}{\epsilon^2 N}\le||\textnormal{sgn}(v)-z||^2\le2(N-|z^H\textnormal{sgn}(v)|)$. At the same time, all the bad fixed points $x$ satisfies
\begin{equation}
\begin{aligned}
\mathring{f}(x)=|x^H\textnormal{sgn}(z)|^2+x^H\mathring{\Delta} x\le \frac{64||\Delta||^2}{\epsilon^2}+\frac{N||\Delta||}{\epsilon}
\end{aligned}
\label{bad1}    
\end{equation}

Combine (\ref{good1}), (\ref{bad1}) with the assumption $||\Delta||\le\frac{\epsilon N}{13}$  and $\alpha<||\Delta||$, we knows that the spectral initialization surpasses all the bad local points. 
\end{proof}

Lemma \ref{lemma:optimalgap} and Lemma \ref{lemma:optimalgapnorm} guaranteed the global optimality of the second order stationary point of problem (\ref{problem:meta}).  Thus in the next theorem, we verify the Hessian $S(x)=\textnormal{ddiag}({C}xx^H)-C$ the Riemann Hessian $S_{\mathcal{R}}(x)=\textnormal{ddiag}(\mathring{C}xx^H)-\mathring{C}$  is P.S.D over $\mathcal{T}_x\mathbb{T}^n$ at the final stationary point. Then we can conclude the global optimality of the converging point of the Generalized Power Method.

\begin{theorem} Given vector $z\in\mathbb{R}^N$. For matrix $C=zz^\top+\Delta$, where $\Delta\in\mathbb{R}^{N\times N}$ is a symmetric perturbation matrix. If $1-\frac{\sqrt{3}}{2}<\epsilon\le\min_{i\in[N]}|z_i|$, $||\Delta||\le\frac{\epsilon^\prime}{28}N$ and $||\Delta ||_\infty\le\frac{\epsilon^\prime}{28}N$, where $\epsilon'=(\epsilon^2+2\epsilon-2)$. When $\alpha\le ||\Delta||$, then the GPM converge to the unique global optimum in the quotient space $\mathbb{R}^n/\sim$. 
\end{theorem}

\begin{proof} From Lemma \ref{lemma:monotonicGPM}, the Generalized Power Method must converge to a stationary point $x$. The fixed point of the generalized power methods satisfies $Sx=0$, which leads to $(\Tilde{C}x)_ix_i=|(\Tilde{C}x)_i|$ and Lemma \ref{lemma:goodstationarypointgpm} guarantees convergence to the good stationary points satisfies $|z^H x|\ge \epsilon N- 4(||\Delta||+\alpha)$. 

From Lemma \ref{lemma:optimalgap}, we also know that, to prove global optimality of $x$, it  suffices to show that $u^HSu>0$ holds for all $u\in\mathbb{C}^n$ such that $u\not=0$ and $u^Hx=0$. This is because
\begin{equation}
    \begin{aligned}
    u^HSu&=\sum_{i=1}^N |u_i|^2|(C_ix)_i|-u^TCu\\
    &=\sum_{i=1}^n |u_i|^2\left||z^Hx|z_i+(\Delta x)_i\right|-|u^Hz|^2-u^H\Delta u\\
    &\ge \sum_{i=1}^n |u_i|^2\left(\epsilon|z^Hx|-|(\Delta x)_i|\right)-|u^H(z-x)|^2-u^H\Delta u\\
    &\ge ||u||^2\left(\epsilon|z^Hx|-||\Delta x||_\infty-||z-x||_2^2-||\Delta||\right)\\
    &\ge ||u||^2\left((2+\epsilon)(\epsilon N- 4(||\Delta||+\alpha))-2N-||\Delta x||_\infty-||\Delta||\right)\\
    &\ge ||u||^2((\epsilon^2+2\epsilon-2)N-(9+4\epsilon)||\Delta||-||\Delta||_\infty)
    \end{aligned}
\end{equation}
Based on the assumption $||\Delta||\le\frac{\epsilon^\prime}{28}N$ and $||\Delta z||_\infty\le\frac{\epsilon^\prime}{28}N$, we know that $u^HSu>0$. 
\end{proof}

\begin{theorem} Given vector $z\in\mathbb{R}^N$ and there exists a constant $\epsilon>0$ such that $\epsilon\le\min_{i\in[N]} |z_i|$. For matrix $C=zz^\top+\Delta$, where $\Delta\in\mathbb{R}^{N\times N}$ is a symmetric perturbation matrix. If $||\Delta||\le\frac{\epsilon}{28}N$ and $||\Delta ||_\infty\le\frac{\epsilon}{28}N$, when $\alpha\le ||\Delta||$, then the normalized GPM converge to the unique global optimum in the quotient space $\mathbb{R}^n/\sim$. 
\label{thm:global}
\end{theorem}
\begin{proof}
From Lemma \ref{lemma:ngpmconverge}, the Normalized Generalized Power Methods must converge to a stationary point $x$. The first order condition of the stationary point is $||\mathring{C} x||_1=x^H\mathring{C} x$ and Lemma \ref{lemma:goodstationarypointngpm} guarantees convergence to the good stationary points satisfies $|z^H x|\ge N- \frac{4}{\epsilon}(||\Delta||+\alpha)$.

Observe that $S_{\mathcal{R}}(x)x=0$, according to Lemma \ref{lemma:optimalgapnorm}, the only thing we need to prove global optimality of the converged stationary point $x$ is to verify the $u^HS_{\mathcal{R}}(x)u>0$ for all $u^Hx=0$.

\begin{equation}
    \begin{aligned}
    u^HSu&=\sum_{i=1}^N |u_i|^2|(\mathring{C}_ix)_i|-u^T\mathring{C}u\\ 
    &=\sum_{i=1}^n |u_i|^2\left||\textnormal{sgn}(z)^Hx|\textnormal{sgn}(z_i)+(\mathring{\Delta} x)_i\right|-|u^H\textnormal{sgn}(z)|^2-u^H\mathring{\Delta} u\\
    &\ge \sum_{i=1}^n |u_i|^2\left(|\textnormal{sgn}(z)^Hx|-|(\mathring{\Delta} x)_i|\right)-|u^H(\textnormal{sgn}(z)-x)|^2-u^H\Delta u-\alpha ||u||_2^2\\
    &\ge ||u||^2\left(|\textnormal{sgn}(z)^Hx|-||\mathring{\Delta} x||_\infty-||\textnormal{sgn}(z)-x||_2^2-||\mathring{\Delta}||\right)\\
    &\ge ||u||^2\left( N- \frac{12}{\epsilon}(||\Delta||+\alpha)-\frac{1}{\epsilon}||\Delta||_\infty-\frac{1}{\epsilon}||\Delta||\right)
    \end{aligned}
\end{equation}
Based on the assumption $||\Delta||\le\frac{\epsilon}{28}N$ and $||\Delta z||_\infty\le\frac{\epsilon}{28}N$, we know that $u^HSu>0$. 
\end{proof}

\subsubsection{Linear Rate Convergence}
\label{subsection:linearrate}

In this section, following \citep{liu2017estimation}, we provide the proof of linear rate convergence of the normalized Generalized Power Method on our problem. With out loss of generality, we assume $1=\arg\min_{\theta\in\{1,-1\}}||\theta y^k-z||_2$ for all $k\in\mathbb{N}$, where $\{y^{k}\}_{k\in\mathbb{N}}$ is the iterates generated by the normalized generalized power method.

\begin{theorem} [Estimation Bound] Suppose that $||\Delta||\le\frac{N\epsilon}{16}$ and $\alpha<\frac{N\epsilon}{6}$. Then the iterates $\{y^{k}\}_{k\in\mathbb{N}}$ generated by the normalized generalized power method satisfies
$$
||y^{k+1}-\textnormal{sgn}(z)||\le \mu^{k+1}||y^0,\textnormal{sgn}(z)||+\frac{\nu}{1-\mu}\frac{8||\Delta||}{\epsilon\sqrt{N}}
$$
for all $k\in\mathbb{N}$, where
$$
\mu=\frac{16({\alpha}+{||\Delta||})}{(7N\epsilon-8{\alpha})}<1, \nu=\frac{2N}{7N-8\frac{\alpha}{\epsilon}}.
$$

\end{theorem}

\begin{proof} This Theorem is a direct adaptation of \citep[Theorem 3.1]{liu2017estimation}.
\end{proof}

Based on the previous estimation bound. We can build the local error bounds to guarantees global convergence. The local error bounds provide an estimation of the distance between any points in $\textnormal{sgn}(z)$'s neighborhood and the global optima of the original optimization problem. To do this, we first define two mappings $\Sigma:\mathbb{T}^n\rightarrow\mathbb{H}^n$ and $\rho:\mathbb{T}^n\rightarrow\mathbb{R}+$ as
$$\Sigma(z)=\textnormal{diag}(|\mathring{C}z|)-\mathring{C},\quad\rho(z)=||\Sigma(z)z||_2.
$$
Then we can have the following results
\begin{lemma} We denote $z^\ast$ the global optimum of problem (\ref{problem:meta}) and $\{y^{(k)}\}_{k\in\mathbb{N}}$ the iterates generated by the Normalized Generalized Power Method. If $\alpha\le||\Delta||\le\frac{\epsilon}{216}N$ and $||\Delta ||_\infty\le\frac{\epsilon}{12}N$, then we have
\begin{itemize}
    \item (Local Error Bound) $||y-z^\ast||\le\frac{N}{4}\rho(y)$ holds for all
    \item$\rho(y^k)\le a||y^{k+1}-y^k||_2$ holds for some constant $a$.
\end{itemize}
\label{lemma:localbound}
\end{lemma}
\begin{proof}
\textbf{Proof of Local Error Bound} To prove the local error bound, we make the following decomposition $||\Sigma(y)y||\ge||\Sigma(z^\ast)y||-||(\Sigma(y)-\Sigma(z^\ast))y||$. We first build the lower bound of $||(\Sigma(y)-\Sigma(z^\ast))y||$ following \citep[Proposition 4.2]{liu2017estimation} 
\begin{equation}
    \begin{aligned}
    ||(\Sigma(y)-\Sigma(z^\ast))y||&=|||\mathring{C}y|-|\mathring{C}z^\ast|||_2\le||\mathring{C}(y-\hat z)||\\
    &\le \sqrt{N}|z^H(y-z^\ast)|+\frac{\alpha+||\Delta||}{\epsilon}||y-z^\ast||\\
    &\le \sqrt{N}|(z^H-z^\ast)^H(y-z^\ast)|+\sqrt{N}|(z^\ast)^H(y-z^\ast)|+\frac{\alpha+||\Delta||}{\epsilon}||y-z^\ast||\\
    &\le \frac{\alpha+5||\Delta||}{\epsilon}||y-z^\ast||+\frac{1}{2}||y-z^\ast||^2
    \end{aligned}
\end{equation}
Similar to Theorem \ref{thm:global}, we can then lower bound $||\Sigma(z^\ast)y||=||\Sigma(z^\ast)\hat y||$ where $\hat y=(I-\frac{1}{n}z^\ast (z^\ast)^H)(y-z^\ast)$ is the projection of $y-z^\ast$ onto the orthogonal complement of $\textnormal{span}(\hat z)$. At the same time
$$
||\hat u||\ge||y-z^\ast||-\left|\left|\frac{1}{n}z^\ast (z^\ast)^H(y-z^\ast)\right|\right|=||y-z^\ast||-\frac{||y-z^\ast||^2}{2\sqrt{N}}
$$
where the last equality is because $||y-z^\ast||^2=2(N-|y^Hz^\ast|)$. Similar to Theorem \ref{thm:global}, we have
\begin{equation}
    \begin{aligned}
    ||\hat y||||\Sigma(z^\ast)\hat y||&\ge \hat y^H\Sigma(z^\ast)\hat y=\hat y^H(\textnormal{diag}(|\mathring{C}z|)-\mathring{C})\hat y\\
    &=\sum_{i=1}^n |\hat y_i|^2\left||\textnormal{sgn}(z)^Hx|\textnormal{sgn}(z_i)+(\mathring{\Delta} x)_i\right|-|\hat y^H\textnormal{sgn}(z)|^2-\hat y^H\mathring{\Delta} \hat y\\
    &\ge ||\hat y||^2\left( N- \frac{12}{\epsilon}(||\Delta||+\alpha)-\frac{1}{\epsilon}||\Delta||_\infty-\frac{1}{\epsilon}||\Delta||\right).
    \end{aligned}
\end{equation}
Thus
$$
||\Sigma(z^\ast)\hat y||\ge \left(||y-z^\ast||-\frac{||y-z^\ast||^2}{2\sqrt{N}}\right)\left( N- \frac{12}{\epsilon}(||\Delta||+\alpha)-\frac{1}{\epsilon}||\Delta||_\infty-\frac{1}{\epsilon}||\Delta||\right).
$$
At the same time $||y-z^\ast||^2\le||y-z^\ast||(||y-z||+||z-z^\ast||)\le\left(\frac{\sqrt{N}}{2}+\frac{4||\Delta||}{\sqrt{N}}\right)||y-z^\ast||$. Combining all the results we get and finally we have
\begin{equation}
    \begin{aligned}
    \rho(z)\ge\left[\frac{N}{2}-\frac{18(||\Delta||+\alpha)}{\epsilon}-\frac{||\Delta||}{\epsilon}\right]||y-z^\ast||.
    \end{aligned}
\end{equation}
Based on the assumptions $\alpha\le||\Delta||\le\frac{\epsilon}{216}N$ and $||\Delta ||_\infty\le\frac{\epsilon}{12}N$, we know that $\rho(z)\ge \frac{N}{4}d(z,\hat z)$
\textbf{Proof of $\rho(y^k)\le a||y^{k+1}-y^k||_2$} By definition of $y^{k+1}$, $\rho(y^k) = ||\textnormal{diag}(|\mathring{C}y^k|)(y^{k+1}-y^k)||\le||\textnormal{diag}(|\mathring{C}y^k|)||_\infty||y^{k+1}-y^k||$. At the same time, we have
\begin{equation}
    \begin{aligned}
    ||\textnormal{diag}(|\mathring{C}y^k|)||_\infty &\le ||zz^Hy^k||_\infty  +\frac{\alpha+||{\Delta}||_\infty}{\epsilon}\\
    &=|z^Hy^k| +\frac{\alpha+||{\Delta}||_\infty}{\epsilon}\le 2N.
    \end{aligned}
\end{equation}
This leads to the estimate $\rho(y^k)\le 2N||y^{k+1}-y^k||_2$

\end{proof}

\begin{theorem} We make the same assumption as Lemma \ref{lemma:localbound}. We further assumes $\mathring{C}\succeq a_0 I$ for some constant $a^\prime>0$, then the normalized generalized power method linearly converge to the global optimum $z^\ast$.
\end{theorem}
\begin{remark} In \citep{liu2017estimation}, the data generating process doesn't ensures the matrix $C$ is P.S.D. Thus \citep{liu2017estimation} should apply a lower bound on $\alpha$ to ensure $\mathring{C}$ is P.S.D. In our case, the matrix is the covariance matrix of a noisy dataset. Thus it is nature have P.S.D. $\mathring{C}$.
\end{remark}

\begin{proof} For $\mathring{C}\succeq a_0 I$, it's obvious to have sufficient ascent $\mathring{f}(y^{k+1})-\mathring{f}(y^k)\ge a_0 ||y^{k+1}-y^{k}||_2^2$ holds for every iteration (\citep[Lemma 8]{boumal2016nonconvex},\citep[Proposition 4.3(a)]{liu2017estimation}). Thus ${f}(y^{k+1})-{f}(y^k)\ge \epsilon a_0 ||y^{k+1}-y^{k}||_2^2$. Before we present the final linear convergence proof, we first prove that $f(z^\ast)-f(y^k)\le a_1 ||y^k-z^\ast||^2$. This is because
\begin{equation}
    \begin{aligned}
    f(z^\ast)-f(y^k)&\le \frac{1}{\epsilon}(\mathring{f}(z^\ast)-\mathring{f}(y^k))\\
    &=(y^k)^H\left(\textnormal{diag}(|\mathring{C}z^\ast|)-\mathring{C}\right)y^k\\
    &=(y^k-z^\ast)^H\left(\textnormal{diag}(|\mathring{C}z^\ast|)-\mathring{C}\right)(y^k-z^\ast)\\
    &\le (||\mathring{C}||+||\mathring{C}||_\infty)||y^k,z^\ast||^2.
    \end{aligned}
\end{equation}

Now we are equipped with all the inequality needed to provide a global convergence proof. According to \citep[Proof of Theorem 4.1]{liu2017estimation}, we knows that the normalized generalized power methods convergence to the global optimum linearly.
\end{proof}

\subsubsection{Generative Models}
\label{subsection:generative}

In this section, we'll discuss how the random data sampled form the linear fixed effect model (also referred to an "interactive fixed-effect model") satisfies the discordant assumptions we used to prove the global optimization results. We assume the outcomes are generated via the following linear factor model
$$
Y_{jt}=\delta_t+D_{jt}\tau+\theta_t^T \mu_j+e_{jt},\qquad \mathbb{E}[e_{jt}|\delta_t,\mu_j,D_{jt}]=0, \textnormal{Var}[e_{jt}|\delta_t,\mu_j,D_{jt}]=\sigma
$$
where $\delta_t$ is the time fixed effect, $\mu_j$ is the unobserved common factors and $\theta_t$ is a vector of unknown factor loadings. $\epsilon_{jt}$ is the unobserved idiosyncratic noise. $\tau$ is the treatment effect we aim to estimate and $D_{jt}$ is the 0-1 variable according to the treatment assignment to unit $j$ at time $t$. In specific, in the pre-treatment period, $D_{jt}=0$ for all $\forall j\in[N], t\in [T]$. Thus the outcome matrix $Y\in\mathbb{R}^{N\times T}$ can be written in the following compact matrix form
$$
Y=\underbrace{\begin{bmatrix} \mu_1^\top& 1 \\ \mu_2^\top& 1 \\ \cdot&\cdot\\ \cdot&\cdot \\ \cdot&\cdot \\ \mu_N^\top & 1 \end{bmatrix}}_{\mu} \underbrace{\begin{bmatrix} \theta_1 &\cdots &\theta_t\\ \delta_1 &\cdots & \delta_t \end{bmatrix}}_{\theta} +W
$$
where $W$ is a matrix whose entries are i.i.d. standard normal random variables denote the measurement noise. We consider the time factor is sampled from a underlying distribution $\begin{pmatrix}\theta_i\\ \delta_i\end{pmatrix}\sim p(\Tilde{\theta},\Tilde{\Sigma})$ and $\Sigma_\theta \triangleq \Tilde{\theta}\Tilde{\theta}^\top+\Tilde{\Sigma} = \mathbb{E}\begin{pmatrix}\theta_i\\ \delta_i\end{pmatrix}\begin{pmatrix}\theta_i\\ \delta_i\end{pmatrix}^\top$. Then we knows that $\mathbb{E}YY^\top = \mu\Sigma_\theta\mu^\top +\sigma I_N$. We first assume that $\Sigma_\theta$ is a non-degenerate covariance matrix.
\begin{assumption}
$\Sigma_\theta$ is positive semi-definite.
\end{assumption}

Assumption \ref{assumption:realizable} means that the matrix $\Sigma\triangleq \mu\Sigma_\theta\mu^\top$ is rank $n-1$. We assume $v=(w_iD_i)_{i=1}^n$ to be vector in the null space of $\Sigma$, where $(w_i,D_i)_{i=1}^n$ is the only  realizable experiment profile in Assumption \ref{assumption:realizable}. To verify that the data generating processing satisfies the assumptions we made for global convergence. We further made the following assumptions to the regularity of the problem.

\begin{assumption}[Regularity of the Problem] We further assume the following regularity properties of the covariance matrix  and random sample $Y_t$
\begin{itemize}
\setlength{\itemsep}{0pt}
\setlength{\parsep}{0pt}
\setlength{\parskip}{0pt}
\item  $||\Sigma^\dagger||\le C_1$ holds for some constant $C_1$.
    \item  $||\Sigma^\dagger||_\infty\le O(N^{c_1})$ holds for some constant $c_1\ge0$.
    \item $||Y_tY_t^\top||\le C_2$ holds almost surely holds for some constant $C_2$.
    \item $||Y_tY_t^\top||_\infty\le O(N^{c_2})$ holds almost surely for some constant $c_2\ge0$.
\end{itemize}
\end{assumption}

\paragraph{Bound $||\sigma N(YY^\top+\sigma I)^{-1}-uu^\top||$} In the following paragraph, we bound the error between the iteration matrix with the rank one ground truth in $\ell_2$ operator norm. To do this, we make the following decomposition
\begin{equation}
    \begin{aligned}
    ||\sigma N(YY^\top+\sigma I)^{-1}-uu^\top||&\le\sigma N||(YY^\top+\sigma I)^{-1}-(\Sigma+\sigma I)^{-1}||+||\sigma N(\Sigma+\sigma I)^{-1}-uu^\top||\\
    &\le \sigma N ||(\Sigma+\sigma I)^{-1}||||YY^\top-\Sigma||+||\sigma N(\Sigma+\sigma I)^{-1}-uu^\top||
    \end{aligned}
\end{equation}

We first bound $||\sigma N(\Sigma+\sigma I)^{-1}-uu^\top||$. To bound this term, we use the geometric series expansion $\frac{\lambda}{\lambda+X}=\sum_{j=0}^\infty (-1)^j\left(\frac{\lambda}{X}\right)^{j+1}$. If $\sigma||\Sigma^\dagger||<1$, then 

\begin{equation}
    \begin{aligned}
    ||\sigma N(\Sigma+\sigma I)^{-1}-uu^\top||\le N \sum_{j=0}^\infty \sigma^{j+1}||\Sigma^\dagger||^{j+1}=\frac{N\sigma||\Sigma^\dagger||}{1-\sigma||\Sigma^\dagger||}
    \end{aligned}
\end{equation}

To bound $ ||\sigma N (\Sigma+\sigma I)^{-1}||||YY^\top-\Sigma||$, we first use the matrix Bernstein inequality \citep{tropp2015introduction,tao2011topics} to bound $||YY^\top-\Sigma||$. We know
$$
||YY^\top-\Sigma||\le \sqrt{\frac{C_2^2\log(\delta)}{T}}+\frac{2C_2\log(\delta)}{T}
$$
with high probability $1-\delta$. At the same time, we have
$$||\sigma N (\Sigma+\sigma I)^{-1}||\le ||\sigma N(\Sigma+\sigma I)^{-1}-uu^\top||+||uu^\top||\le \frac{N}{(1-\sigma||\Sigma^\dagger||_\infty)}.$$ Finally, we achieve
\[
||\sigma N(YY^\top+\sigma I)^{-1}-uu^\top||\le \frac{N\sigma||\Sigma^\dagger||}{1-\sigma||\Sigma^\dagger||}+\frac{N}{(1-\sigma||\Sigma^\dagger||_\infty)}\sqrt{\frac{C_2^2\log(\delta)}{T}}
\]
holds with high probability $1-\delta$.

\paragraph{Bound $||\sigma N(YY^\top+\sigma I)^{-1}-uu^\top||_\infty$} In the following paragraph, we bound the error between the iteration matrix with the rank one ground truth in $\ell_\infty$ operator norm.
\begin{equation}
    \begin{aligned}
    ||\sigma N(YY^\top+\sigma I)^{-1}-uu^\top||_\infty&\le\sigma N||(YY^\top+\sigma I)^{-1}-(\Sigma+\sigma I)^{-1}||_\infty+||\sigma N(\Sigma+\sigma I)^{-1}-uu^\top||_\infty\\
    &\le \sigma N ||(\Sigma+\sigma I)^{-1}||||YY^\top-\Sigma||_\infty+||\sigma N(\Sigma+\sigma I)^{-1}-uu^\top||_\infty
    \end{aligned}
\end{equation}

We first bound $||\sigma N(\Sigma+\sigma I)^{-1}-uu^\top||_\infty$. To bound this term, we use the geometric series expansion $\frac{\lambda}{\lambda+X}=\sum_{j=0}^\infty (-1)^j\left(\frac{\lambda}{X}\right)^{j+1}$. If $\sigma||\Sigma^\dagger||_\infty<1$, then 

\begin{equation}
    \begin{aligned}
    ||\sigma N(\Sigma+\sigma I)^{-1}-uu^\top||_\infty\le N \sum_{j=0}^\infty \sigma^{j+1}||\Sigma^\dagger||_\infty^{j+1}=\frac{N\sigma||\Sigma^\dagger||_\infty}{1-\sigma||\Sigma^\dagger||_\infty}
    \end{aligned}
\end{equation}

To bound $ ||\sigma N (\Sigma+\sigma I)^{-1}||||YY^\top-\Sigma||_\infty$, we first use the matrix Bernstein inequality \citep{tropp2015introduction,tao2011topics} to bound $||YY^\top-\Sigma||_\infty$. We know
$$
||YY^\top-\Sigma||_\infty\le \sqrt{\frac{N^{2c_2}\log(\delta)}{T}}+\frac{2N^{c_2}\log(\delta)}{T}
$$
with high probability $1-\delta$. At the same time, we have
$$||\sigma N (\Sigma+\sigma I)^{-1}||\le ||\sigma N(\Sigma+\sigma I)^{-1}-uu^\top||_\infty+||uu^\top||_\infty\le \frac{N}{\epsilon^2(1-\sigma||\Sigma^\dagger||_\infty)}.$$

We plug in all the bounds and finally get
\[
||\sigma N(YY^\top+\sigma I)^{-1}-uu^\top||_\infty\le \frac{N\sigma||\Sigma^\dagger||_\infty}{1-\sigma||\Sigma^\dagger||_\infty}+\frac{N}{\epsilon^2(1-\sigma||\Sigma^\dagger||_\infty)}\sqrt{\frac{N^{2c_2}\log(\delta)}{T}}
\]
with high probability.

From the discussion in Appendix \ref{subsection:linearrate}, if we can bound both $||\sigma N(YY^\top+\sigma I)^{-1}-uu^\top||$ and $||\sigma N(YY^\top+\sigma I)^{-1}-uu^\top||_\infty$ as $O(\epsilon N)$, then we can have global convergence results. It's easy to check that if we select $\sigma\le \Omega(\epsilon N^{-c_1}), T\ge \Omega(\epsilon^6N^{2c_2})$, then the assumptions for global convergence holds.

\begin{corollary} If $c_1=0$, \emph{i.e.} there exists some constant $C_1$ such that $||\Sigma^\dagger||_\infty\le C_1$, then the noise level $\sigma\le \Omega(\epsilon)$ and {$T>\Omega(\epsilon^6N^{2c_2})$} ensures global convergence of NormSPCD algorithm.
\end{corollary}

\end{document}